\documentclass[sigconf]{acmart}
\pdfoutput=1
\graphicspath{{figures/}}
\usepackage{epstopdf} 
\usepackage{placeins}
\usepackage{balance}
\usepackage[linesnumbered,ruled]{algorithm2e}

\setcopyright{rightsretained}
\begin{document}
%\title{Learning Analytical Query Answers under Dynamic Workloads}
% \title{Dynamically Changing Federated Inferential \\Analytics at the Edge}
\title{Adaptive Learning of Aggregate Analytics under Dynamic Workloads}
% \title{\texttt{Scaffold}: Scaling Federated Analytics under Dynamic Workloads}%temp %temp %temp
\author{Fotis Savva}
\affiliation{%
  \institution{University of Glasgow}
}
\email{f.savva.1@research.gla.ac.uk}

\author{Christos Anagnostopoulos}
\affiliation{%
  \institution{University of Glasgow}
}
\email{christos.anagnostopoulos@glasgow.ac.uk}

\author{Peter Triantafillou}
\affiliation{%
  \institution{University of Warwick}
}
\email{p.triantafillou@warwick.ac.uk}

\begin{abstract}
Large organizations have seamlessly incorporated data-driven decision making in their operations. 
However, as data volumes increase, expensive big data infrastructures are called to rescue. In this setting, analytics tasks become very costly in terms of query response time, resource consumption, and money in cloud deployments, especially when base data are stored across geographically distributed data centers. Therefore, we introduce an adaptive Machine Learning mechanism which is light-weight, stored client-side, can estimate the answers of a variety of aggregate queries and can avoid the big data backend. The estimations are performed in milliseconds are inexpensive and accurate as the mechanism learns from past analytical-query patterns. However, as analytic queries are ad-hoc and analysts' interests change over time we develop solutions that can swiftly and accurately detect such changes and 
adapt to new query patterns. The capabilities of our approach are demonstrated using extensive evaluation with real and synthetic datasets.
\end{abstract}



% \keywords{ACM proceedings, \LaTeX, text tagging}

\maketitle

\section{Introduction}
% \begin{figure}[htbp]
% \begin{center}
% \includegraphics[height=4.5cm,width=5cm,keepaspectratio]{examples-cloud-cost}
% \caption{Costs associated with using the cloud. x-axis is the amount of data used per query and y-axis is the associated costs with the average number of queries daily.}
% \label{fig:clouds-costs}
% \end{center}
% \end{figure}
%Data growing exponentially - BD engines/DBs no longer enough - Additional constraints for interactivity by exploratory analysis processes - exploratory analysis often conducted by extracting statistics from data with queries that have large selectivity and then a much narrower focus - AQP engines are not lightweight they require huge samples - In need of complimentary solutions when interactivity is a must - AQP engines BD engines require high resources often reside in expensive cloud environments - Need something to be in the middle to allow for (1) efficient computation of statistics for exploratory analysis (2) lightweight in terms of storage so that they can reside on a local machine and not in the cloud (3) accurate in terms of answer - Data access is expensive AQP engines relying on data are not lightweight on the other hand queries are ignored but can often provide enough information to estimate statistics.  ( Have previously been used to both guide indexing schemes or other data structures facilitating aggregate estimation. Most recently they've been used along with ML Models to compute aggregates such as COUNT/AVG and Regression Queries. ) 
With the rapid explosion of data volumes and the adoption of data-driven decision making, organizations have been struggling to process data efficiently. Because of that a surge of companies is turning to popular cloud providers that have created large-scale systems capable of storing and processing large data quantities. However, the problem still remains in that multiple analytics queries are issued by multiple analysts (Figure \ref{fig:overview_sys}) which often overburden data clusters and are costly. 
\begin{figure}[htbp]
\begin{center}
\includegraphics[height=6.5cm,width=7cm,keepaspectratio]{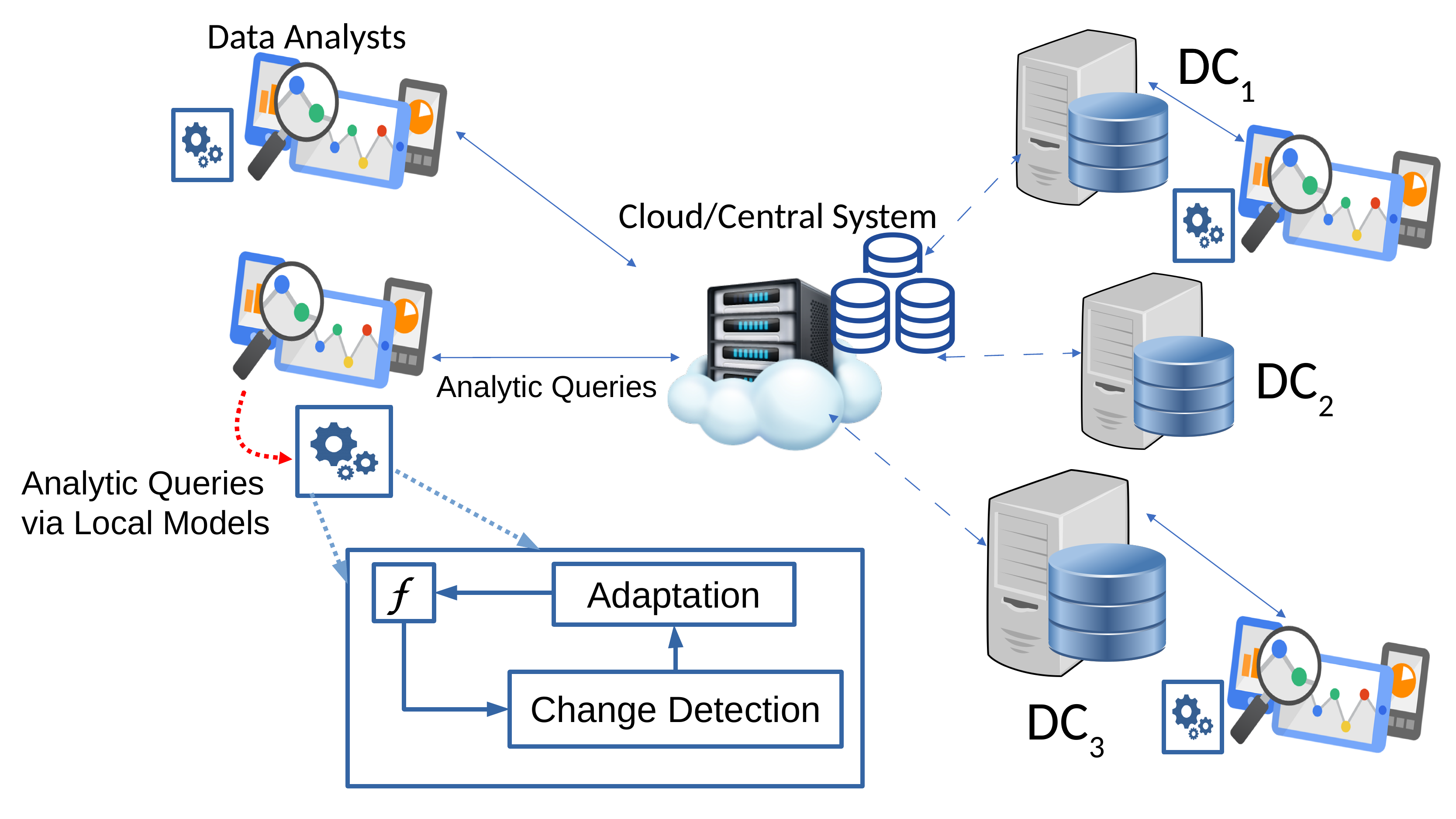}
\caption{Aggregate analytics eco-system with analysts' devices, data centers, and local adaptive ML models.}
\label{fig:overview_sys}
\end{center}
\end{figure}
% For example, looking at Figure \ref{fig:clouds-costs} we can see the exponential costs associated with an increasing size of data. The associated cost is obtained after multiplying the cost of accessing a certain amount of data with the number of queries shown as colored bars. Even though the figure is hypothetical is not far from the truth. This could easily be considered as the daily costs of a medium sized company. It can even exceed these costs if we keep in mind the kind of data analysis and exploration phase, data scientists have to go through before informing a business decision with data-driven logic.
Data analysts should be able to extract information without significant delays so as not to violate the \textit{interactivity constraint} set around 500ms \cite{liu2014effects}. 
Anything over that limit can
negatively affect analysts' experience and productivity. 
This constraint is particularly important in the context of \textit{exploratory analysis} \cite{exploration}. 
Such analyses are an invariable step in the process of understanding data and creating solutions to support business decisions. Furthermore,
 aggregate analytics are becoming increasingly geo-distributed,
which are time consuming and nearly impossible when 
data have to remain federated without the possibility 
of transferring them to central locations \cite{vulimiri2015wanalytics}. 
Same applies to sensitive data that can only be accessed 
via aggregate queries with no data samples allowed.

\textbf{Vision:} Depicted at Figure \ref{fig:overview_sys} is our vision for an aggregate analytics learning \& prediction 
system that is light-weight, stored on an Analyst's Device (AD) and
adaptive to dynamic query workloads.  
This allows the exploratory process to be executed locally at 
ADs providing predictions to aggregate queries 
thus not overburdening the Cloud/Central System (CS).
Prediction-based aggregate analytics is
expected to save computational and communication resources,
which could be devoted to cases where accurate answers to aggregate 
queries are demanded. From the CS's perspective, 
our system acts as a \textit{pseudo}-caching mechanism to 
reduce communication overhead and computational load 
when it is necessary, thus, allowing for other tasks/processes to run.
% A key to these analyses is that an approximate answer to any query is often enough to move forward. As such, over the last few decades research has focused into systems that allow Approximate Query Processing (AQP)\cite{garofalakis2001approximate} to facilitate the process of data analysis with this in mind. By trading off some of the accuracy they allow for order of magnitude speed-ups in execution. 

% Although AQP systems offer a straight-forward and efficient solution to the problem of efficiently processing queries they come at a cost. They require large samples and reside in cloud systems which makes them costly to maintain. In addition, in cases where a number of analysts is using the AQP engine to process queries, the cloud system is overburdened with a number of jobs waiting to be executed at the cluster.In case of no connectivity to the cloud, analysis cannot be performed.Thus, what we propose is a complimentary system to that of AQP engines addressing some of their short-comings and better highlighting their strengths.
%We envision a system that is light-weight and can be stored on an analyst's machine. This allows for the exploratory process to be executed locally at the analyst's machine, without overburdening the cloud system thus saving up resources (and money) to be used in cases where accurate answers are needed. From the cloud provider's standpoint, our solution could act as a caching mechanism to reduce load when it's necessary thus allowing for other operations to run.

Our system offers a learning-based, prediction-driven way 
of performing aggregate analytics in ADs accessing \textit{no data}.
It neither requires data transmission from CS to ADs 
nor from ADs to CS. 
What makes such a system possible is the exploitation of previously executed queries and their answers 
sitting in log files. We adopt Machine Learning (ML) regression models that learn to associate past executed queries with their
answers, and can in turn locally predict the answers 
of \textit{new} queries. 
Subsequent aggregate queries are answered in milliseconds, 
thus, fulfilling the interactivity constraint. 

Furthermore, our framework can directly adapt to analysts' (dynamic) query workloads by monitoring
the analysts' query patterns and 
adjusting their parameters. Shown at Figure \ref{fig:overview_sys} are the ML models $f$ and mechanisms developed for detecting and adapting to changes in query patterns. Both of them are discussed in Sections \ref{sec:cdm} and \ref{sec:adapt}. 

%Transitioning from static environments to dynamic environments concerning query distribution
%Distributed environment with analysts scattered across different locations. After models have been transmitted we have not eliminated the probability that the query patterns for individual analysts will change. So our system will need to address such concerns and adapt the models w.r.t to new and changing query patterns. The challenges are many : from deciding when the models present are not enough to cover the needs of an analyst, to actually training new models on which queries and whether to inform other analysts about this change. 

\textbf{Challenges \& Contribution:}
A large number of analysts exist within an organization with 
diverse analytics interests thus their query patterns are expected to differ, 
accessing different parts of the whole data-space. We are challenged to 
support model training over vastly different patterns, 
which are to be drastically changing or expanding in dynamic environments. 
Moreover the models have to be up-to-date w.r.t. pattern changes, which require early query pattern change detection and efficient adaptation. 
Given these challenges, our contributions are:
\begin{enumerate}
    \item A novel query-driven mechanism and query representation that associates queries with their respective answers and can be used by ML models.
    \item A local change detection mechanism for detecting changes in query patterns based on our prediction-error approximation.
    \item A reciprocity-based adaptation mechanism in light of novel query patterns, which efficiently engages the CS to validate the prediction-to-adaptation states transition and guarantees system convergence.
    \item Comprehensive assessment of the system performance and sensitivity analysis using real and synthetic data and query workloads. 
\end{enumerate}
% The rest of the sections are organized as follows : Section \ref{sec:query-driven} outlines the query-driven learning methodology which is used to train ML models that predict the answers of aggregate queries. Sections \ref{sec:cdm} describes the process of detecting when dynamic workloads change indicating the need to adapt ML models. Finally, Section \ref{sec:adapt} shows how models are adapted and how this process can be expedited by taking advantage of queries issued by multiple analysts.
\section{Fundamentals of Query-Driven Learning}
\label{sec:query-driven}
The fundamentals of query-driven mechanism for analytics are: (1) transforming analytic queries in a real-valued vectorial space, (2) quantization of vectorial space, extracting query patterns and (3) training of local regression models for predicting query answers using past issued queries. Principally, we learn to associate the result of a query using the derived query patterns and linking these patterns with local regression models. 
Given an unseen query, we project it to the closest query pattern we have learned and then 
predict its corresponding result without executing the query over the data in a DC/CS. 
%Figure for viewing DB as black-box
% To \textit{relational} systems one have to perform multiple joins to get more attributes but in the end is the \textit{selection predicates} that would make the difference in the final result. %Figure with normalized and denormalized tables.
% \begin{figure}
% \begin{center}
% \includegraphics[height=4.5cm,width=5cm,keepaspectratio]{db-storage}
% \caption{Vectorial data stored in different formats which has no effect on the aggregate result returned by a query.}
% \label{fig:storing-data}
% \end{center}
% \end{figure}
\begin{definition} 
A dataset $\mathcal{B}= \{\mathbf{a}\}_{i=1}^n$ is a set of $n$ row data vectors $\mathbf{a} = [a_1,\ldots, a_d] \in \mathbb{R}^{d}$ with real attributes $a_{i} \in \mathbb{R}$.
\end{definition}

Analytics queries are issued over a $d$-dimensional data space and bear two key characteristics: First, they define a subspace of interest, using various predicates on attribute (dimension) values. Second, they perform  
aggregate functions over said data subspaces (to derive key statistics over the subspace of interest). 
%Both \textit{non-relational} and \textit{relational} database systems can be considered as collections of data vectors either grouped in collection of relations or being part of datasets. We can store data in either of the two settings while the aggregate result of queries is the same. 
% Figure \ref{fig:storing-data} shows an example in which data vectors are stored under different formats without any difference to the aggregate results produced; it is just the way of performing data manipulation/aggregation that differs.
We adopt a general vectorial representation for modeling a query over any type of data storage/processing system. 

Predicates over attributes define a data subspace over $\mathcal{B}$ formed by a sequence of logical conjunctions using (in)equality constraints ($\leq, \geq, =$). A \textit{range-predicate} restricts an attribute $a_i$ to be within range [$l_{i}$, $u_{i}$]: $a_i \geq l_{i} \land a_i \leq u_{i}$. We model a range query over $\mathcal{B}$ 
through conjunctions of predicates, i.e., $\bigwedge_{i=1}^{d}(l_{i} \leq a_{i} \leq u_{i})$ represented as a vector in $\mathbb{R}^{2d}$. 
% \subsection{Extending to other operators}\textit{To-Go in Appendix}
% Using range-predicates with the formalism in Def. \ref{def:predicate} a wide range of queries can be supported as shown in the example. The only thing remaining is to distinguish between various operators such as $(\leq, <, \geq, >, =)$ and various sequences of such predicates using conjuction/disjunction/negation. First of all, it is straightforward to derive appropriate value to transform various equality operators to more general ones such as $<$ to $\leq$ etc. We simply decrease the number on the right-hand side by one $x < 10 = x \leq 9$. As for the sequence of negations/conjuctions/disjunctions we rely on enough queries being present for the ML algorithm to be able to distinguish between such operators. As this is out of the scope of our current work we leave this for future work.
%Structural information incorporation of another meta attribute to represent conjunction/disjunction/negation between predicates. If existence of multi-predicates on same attribute, we could find way to estimate which one is more important and leave that one. We should not worry about queries making no sense and we should also understand that the general purpose of this project is to increase efficiency on most common queries and not outlier queries appearing once.
\begin{definition}
A (range) row query vector is defined as $\mathbf{q} = [l_{1},u_{1}, \ldots, l_{d}, u_{d}] \in \mathbb{R}^{2d}$ corresponding to the range query $\bigwedge_{i=1}^{d}(l_{i} \leq a_{i} \leq u_{i})$. The distance between two queries $\mathbf{q}$ and $\mathbf{q}'$ is defined as the $L_{2}^{2}$ norm (Euclidean distance): $\lVert \mathbf{q}-\mathbf{q}' \rVert^2_{2} = \sum_{i=1}^{d}(l_{i}-l'_{i})^{2}+(u_{i}-u'_{i})^{2}$.
\end{definition}

This representation is flexible enough to accommodate a wide variety of queries. As the dimensionality of the query vector is proportional to the data vector, queries with predicates bounding the values of different attributes can be used by the same ML algorithm. This means that only a number of $(l_i, u_i)$ values are set w.r.t the number of predicates for a given query. In addition, we make no assumptions as to the back-end system as what is being parsed are the filters in a query. This allows the mechanism to be deployed in parallel to multiple analytic systems. 
In query-driven learning, we learn to associate a query with its corresponding aggregate result (a scalar $y \in \mathbb{R}$). This is achieved using a $training$ set of query-result pairs $\mathcal{T} = \{(\mathbf{q}_i, y_i)\}_{i=1}^N$ 
obtained after executing $N$ queries over dataset $\mathcal{B}$. 
The goal is to develop an ML model based on $\mathcal{T}$ to 
minimize the expected prediction error between actual $y$ and predicted $\hat{y}$, $\mathbb{E}[(\hat{y}-y)^{2}]$ and predict the result of any unseen query \textit{without} executing it over $\mathcal{B}$.
%\end{definition}

% \subsection{Restricting dimensionality of the representation}
% %blinkdb focused on certain most frequent columns
% %spatio-temporal analysis is restricted on two columns
% %data canopy refer to sdss 55% of queries target non-disticnt columns
% As the number of columns/attributes/dimensions gets larger our representation will be moving towards a high-dimensional space causing problems to our underlying models. One way to tackle this is to use unsupervised dimensionality reduction techniques \cite{roweis2000nonlinear} which will reduce the dimensionality of our given vectors. Another more straightforward way to tackle this is to restrict the number of columns we built our models over. As examined queries focus on a subset of the original column set \cite{agarwal2013blinkdb,wasay2017data}, allowing us to use a heuristic to build models only over the most frequently used columns. Especially in the case of spatio-temporal datasets the focus is usually on the spatial and temporal dimensions to which an analyst applies filters and then examines descriptive statistics over other columns.
\subsection{Query Space Clustering}
%Parameter values issued by different users often form clusters (SDSS graph) ref.
%Users are potentially interested in vastly different regions of the dataset thus issuing queries with different parameter values
%Clustering is a machine learning method that allows for this grouping of similar parameter values into disjoint sets. This would allow for more accuracy in the models trained subsequently and more would allow for the transmission of more approriate models to user's who have executed queries similar to those belonging to the clsuter associated with the model.
%Explain the Process of clustering. Having a log filled with queries and their different parameters. We hold their meta-vectors.
%Clustering is then, and give the objective function for clustering algorithms. Can use any kind of off-the-shelf and thoroughly tested clustering algorithms (references). Clustering results in a number of cluster-heads (either pre-defined or given by the algorithm). Each cluster-head has a number of queries associated with it. Figure for clustering showing Voronoi cells
%Then every query is associated with a cluster. Thus we end up with k disjoint sets. Those k disjoint sets are then used to train Supervised Machine Learning algorithms that can predict the results of new queries.
%Analysts extract information from large datasets by issuing queries restricting the data-space and returning descriptive statistics. 
Recent research analyzing analytics workloads from various domains has shown that queries within analytics workloads share patterns and their results are \textit{similar} having various degrees of overlap \cite{wasay2017data}. 
% Figure \ref{fig:clusters} shows the formation of clusters for two different query parameters from a real query workload obtained from Sloan Digital Sky Survey (SDSS)\cite{szalay2002sdss}. 
% It is shown that the queries form clusters around certain parameter values depicting the interests of analysts over  underlying data sub-spaces.
Based on this evidence, we mine query logs (the $training$ set) and discover clusters of queries (in the vectorial $d$-dimensional space), having similar predicate parameters.
% \begin{figure}
% \begin{center}
% \includegraphics[height=4cm,width=4cm]{clusters-queries}
% \caption{Parameters forming clusters in a real workload (Source SDSS \cite{szalay2002sdss})}
% \label{fig:clusters}
% \end{center}
% \end{figure}
 This partitioning is fundamental to get accurate ML models for predictive analytics, as we then associate different ML predictive models with different clusters.
In this way, learning different data sub-sets is proven to be more efficient in terms of \textit{explainability / model-fitting} and \textit{predictability} than having one global ML model learn everything and is also  \cite{masoudnia2014mixture} known as \textit{ensemble learning} \cite{friedman2001elements}.
%Reference to ED architecture
% Therefore, the models transmitted to the analyst's local machine from the ED are only the ones that were trained with subsets of queries that included the user's issued queries to that ED. Hence, such models \textit{capture} the exploratory interests of the ED's analysts analyzing the data subspaces by issuing analytics queries. 
% \begin{table}
% \caption{Variance for $\mathbf{m}$ and result $y$ is reduced by the use of Local models.(Results from experimental workload described in Experimental Evaluation)}
% \begin{tabular}{ |c|c|c| } 
%  \hline
%  & \textbf{Global Model $\mathcal{M}$} & \textbf{Local Models $\mathcal{M}_{k}$} \\ \hline
%  \textbf{Variance $\mathbf{m}$} & $0.124$ &  $0.0017$ \\ \hline
%  \textbf{Variance $y$} & $1.5\cdot10^{10}$ & $3.5\cdot10^8$ \\ \hline
% \end{tabular}
% \label{tab:variance}
% \end{table}
To put this in context, consider a discrete time domain $t \in \mathbb{T} = \{1, 2, 3, \ldots\}$, where at each time instance $t$ an analyst issues a query $\mathbf{q}_{t}$. The query is executed and an answer $y_{t}$ is obtained, forming the pair $(\mathbf{q}_{t},y_{t})$. 
The issued queries are stored in a growing set $\mathcal{C}_{t} = \{(\mathbf{q}_1, y_1),\ldots, (\mathbf{q}_t, y_t)\} = \mathcal{C}_{t-1} \cup \{(\mathbf{q}_{t},y_{t})\}$. Given this set, we incrementally extract knowledge from the query vectors and then train \textit{local} ML models that predict the associated outputs given new, unseen queries. This is achieved by on-line partitioning the vectors $\{\mathbf{q}_{1}, \ldots, \mathbf{q}_{t}\} \in \mathcal{C}_{t}$ into disjoint clusters that represent the query patterns of the analysts (fundamentally, within each cluster the queries are much more \textit{similar} than the queries in other clusters). The distance between queries quantifies how close the predicate parameters are 
in the vectorial space. Close queries $\mathbf{q}$ and $\mathbf{q}'$ are grouped together into $K$\footnote{The number of clusters $K$ is automatically identified by the clustering algorithm used. \cite{growingnetworks}} clusters w.r.t. $\lVert \mathbf{q} - \mathbf{q}' \rVert^2_{2}$. The objective is to minimize the expected quantization error $\mathbb{E}[\min_{k =1,\ldots,K}\lVert \mathbf{q}-\mathbf{w}_{k}\rVert_{2}^{2}]$
of all queries to their closest cluster \textit{representative} $\mathbf{w}_{k}$, which reflects the analysts query patterns and best represents each cluster. 
The $K$ query representatives $\mathcal{W} = \{\mathbf{w}_{1}, \ldots, \mathbf{w}_{K}\}$ optimally quantize $\mathcal{C}_{t}$ minimizing the expected quantization error while each query $\mathbf{q}$ is projected onto its closest representative
$\mathbf{w}^{*} = \arg\min_{\mathbf{w} \in \mathcal{W}}\lVert \mathbf{q}-\mathbf{w}\rVert^2_{2}$. 
Based on the partitioning of $\mathcal{C}_{t}$, 
we produce $K$ query-disjoint sub-sets such that $\mathcal{C}_k\cap\mathcal{C}_l \equiv \emptyset$ for $k\neq l$ and $\mathcal{C}_{k} =\{(\mathbf{q},y) \in \mathcal{C}_{t} | \mathbf{w}_{k} = \arg\min_{\mathbf{w} \in \mathcal{W}}\lVert \mathbf{q}-\mathbf{w}\rVert_{2}\}$. A \textit{local} ML model is then trained over each subset using the pairs in $\mathcal{C}_{k}, k \in [K]$.
\subsection{Query-Answer Predictive Model}
% \label{sec:query-answer}
% \subsection{Predictive Modeling}
%Data is partitioned - We can use each partition to train separate models - What's the task, minimizing squared error loss, multitude of ML algorithms that can do that following different techniques. Result is an ML model which can in turn be transmitted to the user's/analysts local machine. A number of models can be trained and transmitted w.r.t the query parameters the user has issued so far, and the kind of aggregates. For each aggregate we train different model. The analyst uses the ML models to avoid executing queries in the cloud - Gets back estimated results along with an estimation of the error.(Figure, x - parameter, y- result, shade error, different color lines aggregates). Error is the predicted generalization error given by the algorithm can be used to construct error bounds.
Each aggregate result $y$ from the pair $(\mathbf{q},y) \in \mathcal{C}_{t}$ 
represents the $exact$ answer produced by the CS. Essentially, $y$ is produced by an unknown function $g$ which we wish to learn. Such function produces query-answers w.r.t an unknown distribution $p(y|\mathbf{q})$. Our aim is to approximate the \textit{true} functions $g$ for aggregate functions (descriptive statistics) e.g., count, average, max, sum etc. Regression algorithms are trained using query-answer pairs from $\mathcal{C}_{t}$ to minimize the expected prediction error between actual $y = g(\mathbf{q})$, from the true function $g$, and predicted $\hat{y}$ from an approximated function $\hat{g}$, i.e., $\mathbb{E}[(g(\mathbf{q})-\hat{g}(\mathbf{q}))^{2}]$.
%The objective in \textit{Eq.}\ref{eq:obj-error} is to minimize the Mean Squared Error (MSE)\footnote{Note, various supervised learning algorithms minimize a variant of the objective in (\ref{eq:objective}) including a regularization parameter as a technique to avoid \textit{overfitting} \cite{friedman2001elements}.}
%\begin{equation}
%\frac{1}{t}\sum_{\tau=1}^t(y_{\tau} - \hat{g}(\mathbf{q}_{\tau}))^2.
%\label{eq:objective}
%\end{equation}
After having partitioned the query space into clusters $\mathcal{C}_1,\ldots, \mathcal{C}_K$, we therein train $K$ local ML models, $\mathcal{M}=\{\hat{g}_1,\ldots, \hat{g}_K\}$ that associate queries $\mathbf{q}$ belonging to cluster $\mathcal{C}_{k}$ with their outputs $y$. Each ML model $\hat{g}_{k}$ is trained from query-response pairs $(\mathbf{q},y) \in \mathcal{C}_{t}$ from those queries $\mathbf{q}$ which belong to $\mathcal{C}_{k}$ such that $\mathbf{w}_{k}$ is the closest representative to those queries. The originally trained ML models in DC/CS are then sent to ADs to be used for predicting answers. %of new unseen queries \textit{without} communication and execution to the DC/CS.
Given a query $\mathbf{q}$ only the most representative model $\hat{g}_{k}$ is used for prediction, corresponding to the closest $\mathbf{w}_{k}$:
\begin{eqnarray} 
\hat{y} & = & \sum_{k=1}^{K}\mathcal{I}_{k}\hat{g}_{k}(\mathbf{q}) 
\end{eqnarray}
where $\mathcal{I}_{k} = 1$ if $\mathbf{w}_{k} = \arg\min_{\mathbf{w} \in \mathcal{W}}\lVert \mathbf{q}-\mathbf{w}\rVert^2_{2}$; 0 otherwise. 
\section{Query Pattern Change Detection}
\label{sec:cdm}
Suppose that all trained ML models $\{\hat{g}_{k}\}_{k=1}^{K}$ are delivered to the analysts from CS, indicating that 
the mechanism enters its \textit{prediction mode}. 
That is, for each incoming query, it predicts the answer and delivers it
back to the analysts \textit{without} query execution. If we assumed a
stationary query pattern distribution, in which queries and analysts' interests do not change, then this would suffice. 
However, this is not realistic as it is highly likely that \textit{analysts interests} change over time (e.g., during exploratory analytics tasks, which are considered as ad-hoc processes \cite{exploration}).
So, dynamic workloads may render the $\{\hat{g}_{k}\}_{k=1}^{K}$ models obsolete, as they were trained using \textit{past} query patterns following distributions which may now be different. 
Accommodating such dynamics %rendring derive ML models obsolete, are not usually addressed, as the main concern typically is to build models and not how to maintain and monitor them. However, this 
is becoming increasingly important as ML is widely adopted in software in production \cite{sculley2015hidden}. Specifically,
when referring to analysts' interests, we refer to analysts who are tasked with informing different business decision processes. If those tasks change, the data subspaces to be analyzed become different, 
which results in changed query patterns. 
%Adapting to such cases is crucial if we desire our local models to remain accurate. 
If models cannot be adaptive, expected prediction errors can become arbitrarily high. 
%as the initial ML models were trained to approximate the original distribution $p(y|\mathbf{q})$. 
When $p(y|\mathbf{q})$ changes to $p'(y|\mathbf{q})$, it is highly likely that any previous approximation would produce high-error answers, unless $p(y|\mathbf{q}) \approx p'(y|\mathbf{q})$. 
We capture such dynamics as \textit{concept drift} \cite{tsymbal2004problem,ditzler2015learning} --  
many methods have been developed for adjusting when this arises \cite{gepperth2016incremental,ditzler2015learning}.

We introduce a Change Detection Mechanism (CDM) and an Adaptation Mechanism (ADM) (shown at Figure \ref{fig:overview_sys}) addressing this concern raising a number of challenges:
(1) How to detect a query pattern change; we need to enable triggers that alert the mechanism being in prediction mode in case of a \textit{concept drift};
(2) What kind of action should we take in case that happens, i.e., what strategy to follow for updating the ML models; (3) How should we notify users, analysts, and applications about such change(s) or even who to notify; shall we transmit an update to all users or just the affected ones? 
We explore these challenges and the  describe the decisions we take in tackling them in the remainder.

\subsection{Change Detection Mechanism}
So far, we have trained $K$ different local ML models to 
predict answers involving only the $k$-th model that best represents 
a new incoming query through the representative $\mathbf{w}_{k}$. 
This %ensemble-centric prediction 
requires to 
individually monitor whether the query
representatives, used for prediction via their respective models, are \textit{still} representatives in long-term predictions
or whether the analysts' query patterns have changed. 
In this case, we need to introduce a CDM that triggers when the original query representative has significantly diverged from the estimated one. 

Our approach can be best understood by first assuming that the CDM maintains an on-line average of the prediction error $(y-\hat{y})^2$ such that : $u_{k} \approx \mathbb{E}[(y-\hat{y})^{2}|\mathbf{q}]$. This is done for each query representative $\mathbf{w}_{k}$ across different users. Should the expected error $u_{k}$ escalate significantly, then this may signal that a query pattern has shifted around the `region' represented by the representative $\mathbf{w}_{k}$.
But, recall that during the prediction mode, the actual $y$ is unknown since our goal is to predict accurate answers but without executing the query itself. Hence, we develop an approximation mechanism for change detection, not requiring query executions over CS/DC.

Once we have trained the individual ML models $\mathcal{M}$ and calculated their expected prediction accuracy (using an independent test sample drawn from the original set of queries) we obtain the Expected Prediction Error (EPE), which will be constant across all possible queries associated with a particular query representative defined as: $EPE = \mathbb{E}\Big[\big(g(\mathbf{q}) - \sum_{\kappa=1}^{K}\mathcal{I}_{\kappa}\hat{g}_{\kappa}(\mathbf{q})\big)^{2}\Big]$. Using the EPE, we wish to find a fine-grained estimate of the true prediction error rather just assuming this is constant for each and every unseen query.
% (The latter is problematic as when a pattern shift occurs 
% increasing the true prediction error, the estimate remains unchanged and untrustworthy).
%Figure for correlation of true prediction error and ours
%Figure for increase of error when pattern shift occurs (?)

To do this, we have analyzed the error behaviour under changing query patterns. 
Our findings reveal an interesting fact: The Euclidean distance $d(\mathbf{q},\mathbf{w}_{k}) = \lVert \mathbf{q}-\mathbf{w}_{k} \rVert^2_{2}$ of a random query $\mathbf{q}$ from its closest query representative $\mathbf{w}_{k}$ is strongly 
correlated with the 
associated prediction error $(y-\hat{y})^{2}$.\footnote{A 0.3 Pearson's Correlations was obtained on a real dataset.} 
Considering the correlation between $d(\mathbf{q},\mathbf{w}_{k})$ and the local $u_{k}$, we define a distance-based prediction error $\tilde{u}_{k}$ of a query $\mathbf{q}$ as:
\begin{equation}
\tilde{u}_i = \ln{(1 + d(\mathbf{q},\mathbf{w}_{k})-\min_{\mathbf{q}_{\ell} \in \mathcal{C}_{k}} d(\mathbf{w}_{k},\mathbf{q}_{\ell}))\cdot u_{k}},
\label{eq:error}
\end{equation}
where the natural-log operator acts as a penalizing/discount factor for queries given their distance from the closest representative $\mathbf{w}_{k}$. 
The second term within the natural-log operator, $\min_{\mathbf{q}_{\ell} \in \mathcal{C}_{k}} d(\mathbf{w}_{k},\mathbf{q}_{\ell})$ is the minimum distance between the query representative $\mathbf{w}_{k}$ and the associated 
queries $\mathbf{q} \in \mathcal{C}_{k}$. 
We subtract the minimum distance from $d(\mathbf{q},\mathbf{w}_{k})$ so that the scale of 
the numbers will not affect the computation of the error. 

% For each local ML model $\hat{g}_{k}$, the local EPE $u_{k}$ is constant and is put in place so as to express the $\tilde{u}_{k}(\mathbf{q})$ in terms of prediction error and not just based on the distance $d(\mathbf{q},\mathbf{w}_{k})$ from the closest representative. Hence, the analyst has the flexibility to provide a threshold w.r.t. their acceptable prediction error for data analysis. 
We base our novel CDM in (\ref{eq:error}) using the series of error approximations $\{\tilde{u}_{t}\}$ for monitoring concept drifts in query patterns during prediction mode without executing the queries.

Consider the incoming unseen (random) queries $(\mathbf{q}_{0}, \mathbf{q}_{1}, \ldots, \mathbf{q}_t)$  arriving in a sequence $t \in \mathbb{T}$ during prediction mode. They are answered  by specific local ML models $(\hat{g_{0}}, \hat{g_{1}}, \ldots,\hat{g}_{k} )$, generating a series of distance-based error estimations $\{\tilde{u}_{t}\}$, $t \in \mathbb{T}$. The CDM monitors this series and, based on a specific threshold, signals 
the existence of \textit{concept drift}, i.e., checks whether the probability distribution of the queries has changed. Based on the series of error estimations, we learn two query distributions: (1) the \textit{expected} query distribution, which is represented by the query representatives and (2) the \textit{novel} query distribution, which cannot be represented by the current query representatives. The expected distribution $p_{0}(\tilde{u})$ is estimated given a training period from $\tilde{u}_{k}$ values corresponding to queries with closest representative $\mathbf{w}_{k}$. 
The novel distribution $p_{1}(\tilde{u})$ is estimated from $\tilde{u}_{m}$ values corresponding to error values derived from the \textit{rival} representatives $\mathbf{w}_{m}$ of queries with closest $\mathbf{w}_{m}$ and $k \neq m$. Based on this, we estimate the distribution of the error values generated from representatives which were not the closest to the queries, thus, approximating novel error values. Both distributions were approximated by fitting the $p(\tilde{u}) \sim \Gamma(e_{1},e_{2})$ distribution with scale $e_{1}$ and shape $e_{2}$. 

Given a $\tilde{u}_{t}$ value, we calculate the likelihood ratio $s_{t} = \log \frac{p_{1}(\tilde{u}_{t})}{p_{0}(\tilde{u}_{t})}$ and the cumulative 
sum of $s_{t}$ up to time $t$, 
%\begin{eqnarray}
$U_{t} =  \sum_{\tau=0}^{t}s_{\tau}.$
%\label{eq:U}
%\end{eqnarray}
Based on the sequential ratio monitoring for a progressive concept drift in 
distribution \cite{grigg2003use} from $p_{0}$ to $p_{1}$, a decision function is introduced for signaling a potential concept drift expressed as: 
\begin{eqnarray}
G_{t} & = & U_{t} - \min_{0 \leq \tau \leq t}U_{\tau-1}.
\label{eq:G}
\end{eqnarray}
The decision function in (\ref{eq:G}) indicates the current cumulative sum of ratios minus its current minimum value. This denotes that the \textit{change time} estimate $t^{*}$ is the time following the current minimum of the cumulative sum, i.e., 
%\begin{eqnarray}
$t^{*}  =  \arg \min_{0 \leq \tau \leq t}U_{\tau}.$
%\label{eq:tstar}
%\end{eqnarray}
Therefore, given that $U_{t} = U_{t-1} + s_{t}$, the decision function in (\ref{eq:G}) this is re-written in a recursive form:
%\begin{eqnarray}
$G_{t}  =  \{G_{t-1} + s_{t}\}^{+}$
%\label{eq:G+}
%\end{eqnarray}
with $\{z\}^{+} = \max(z,0)$ setting, by convention, $U_{-1} = 0$ and $G_{-1} = 0$. Hence, a concept drift of query patterns projected over 
the query representatives space is detected at time $t_{D}$:
%\begin{eqnarray}
$t_{D}  =  \min\{t \geq 0: G_{t} > h\}.$
%\label{eq:tD}
%\end{eqnarray}
The parameter $h$ is usually set $3\sigma \leq h \leq 5\sigma$ with $\sigma$ the standard deviation of $\tilde{u}$. The process is shown at Figure \ref{fig:buff-phase}, the cumulative sum of ratios exceeds the threshold $h$ as soon as queries are issued from an unknown distribution as the error estimates become steadily larger and are not just random fluctuations in errors.
It is worth noting that the change in query distribution is based on fusing the distance between the queries and their closest representatives scaled with the expected prediction error. We refer to this as an indication of degradation in the performance of the model. Given that a change has been detected,
the CDM signals the ADM which transits from \textit{prediction mode} to \textit{buffering mode} as shown at Figure \ref{fig:overview-cdm-adm}. As soon as a change is detected the CDM signals the ADM component, that new query patterns have been detected. In turn, the ADM signals the \textit{Prediction Component} (containing the $\mathcal{M}$ and $\mathcal{W}$) to be put in \texttt{BUFFERING} mode since the prediction component can no longer provide reliable answers for \textit{all} queries.  However, the AD can still leverage the complete system to ask queries following the already known distributions with \textit{only} queries following the new shifted distribution being executed at the CS. By entering \texttt{BUFFERING} mode our ADM starts to adjust for the new query patterns under the AD until \textit{converging}. At that point it signals the \textit{Prediction Component} to switch back to \texttt{PREDICTION} mode, resuming normal operation.
%CUTOFF FOR KDD  

\begin{figure}
\begin{center}
\includegraphics[height=4cm,width=6cm]{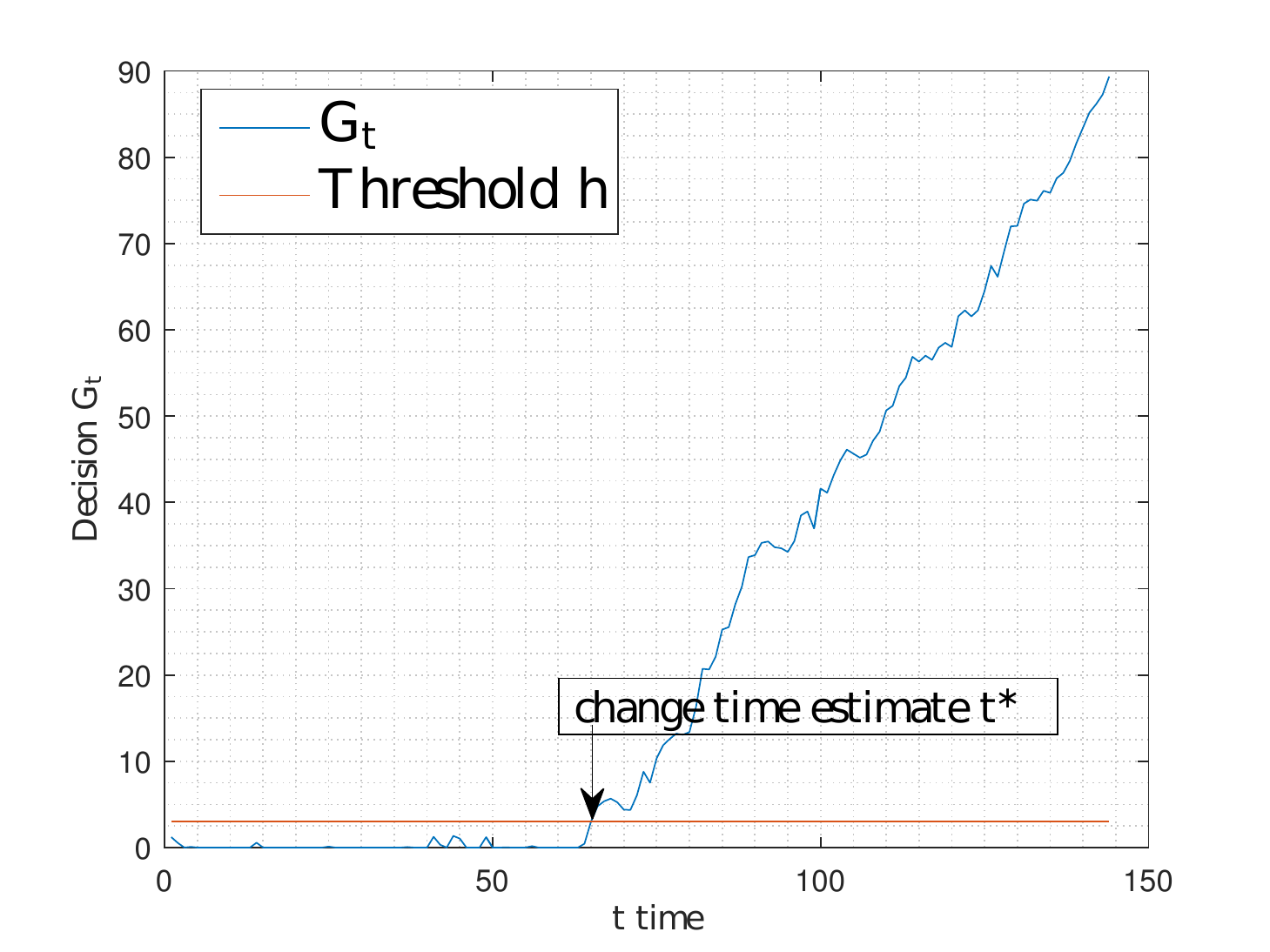}
\caption{Change detection based on the likelihood ratio of the distance-based error triggered when query patterns are shifted.}
\label{fig:buff-phase}
\end{center}
\end{figure}

\begin{figure}
\begin{center}
\includegraphics[height=4cm,width=6cm]{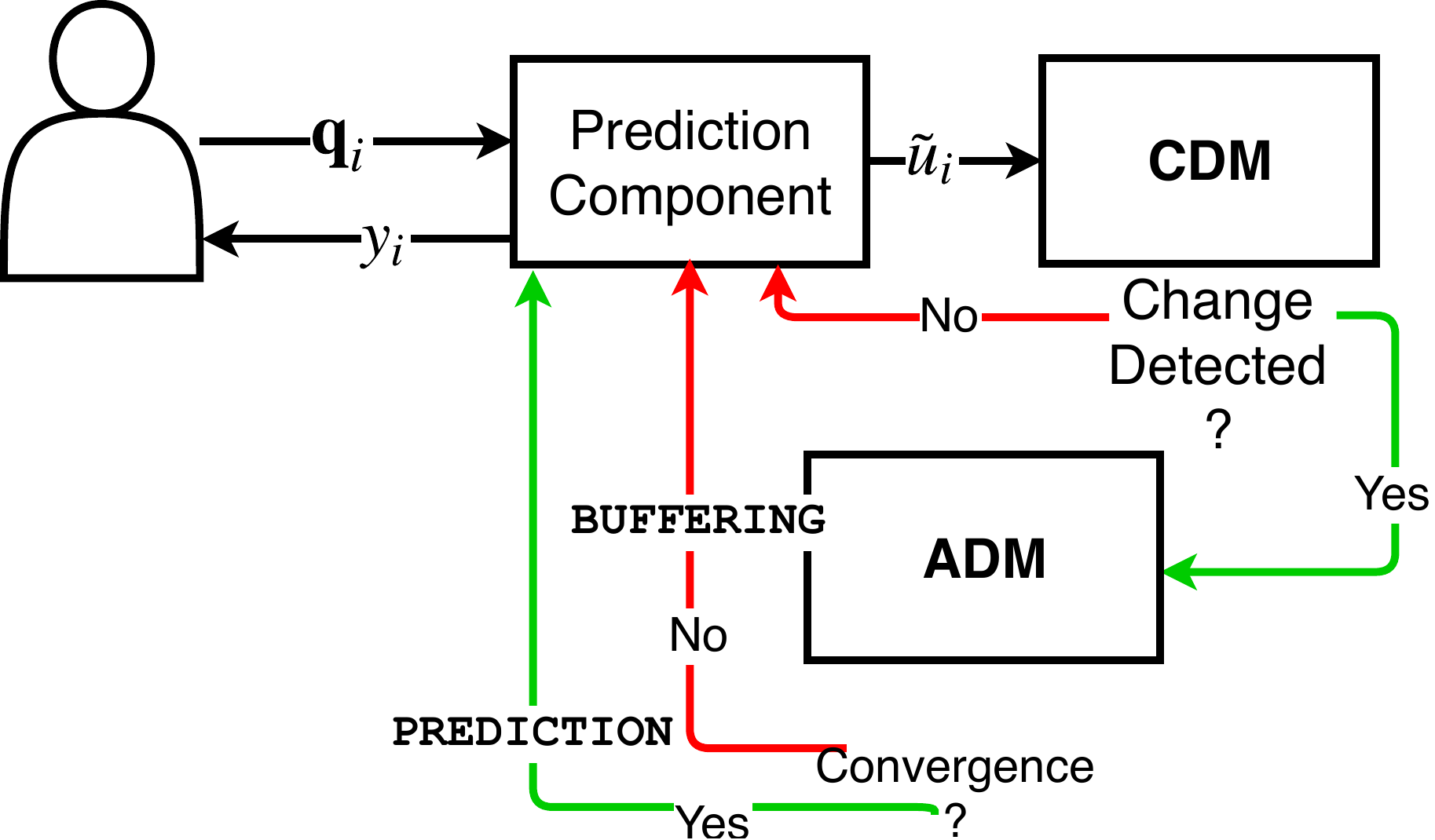}
\caption{Overview of the operation of CDM and ADM which control when estimations can be reliably given to the analyst}
\label{fig:overview-cdm-adm}
\end{center}
\end{figure}
\section{Model Adaptation}
\label{sec:adapt}
In this section, we explain the fundamentals of the ADM along with unintended 
results we can exploit. Once a local model $\hat{g}_k$ transits to buffering mode, 
it is deemed unreliable to accurately predict the answers to incoming queries. 
Therefore, during this phase, queries should be executed 
and their actual answers returned to analysts 
while also being used for adapting the 
model and representative $(\hat{g}_k, \mathbf{w}_k)$. In the beginning of  buffering mode, 
$\hat{g}_k$ and $\mathbf{w}_k$ are sent to CS, 
which will tune/adapt their parameters using 
the actual execution of queries. 
% After the convergence of $\mathcal{M}_k$ and $\mathbf{w}_k$, the CS 
% sends back the updated parameter to the AD which are enhanced from the execution of queries coming from that AD and possible affiliates $\mathcal{A}$. Nonetheless, 
To reduce the expected number of queries executed in CS 
during buffering mode, we introduce 
a \textit{query execution selectivity} mechanism 
based on the current estimated error in (\ref{eq:error}).
Specifically, there would still be some queries issued by 
the AD that could be locally answered by 
current models cached in 
AD during that phase. 
Therefore, the AD still monitors incoming queries  and \textit{discriminates}
between two types: (1) the ones that can be locally answered by models in $\mathcal{M}$ and (2) the ones that cannot be answered, since these queries are not well represented by the cached query representatives. The latter queries are then forwarded to 
CS for execution. The selectivity mechanism relies on the 
following rule: an incoming query $\mathbf{q}_{t}$ at an AD 
in buffering mode is locally answered if the distance from 
the query representative of the new query patterns, notated by $\mathbf{w}_{K+1}$, is not the closest representative  i.e., $\mathbf{w}_{k^{*}} = \arg \min_{\mathbf{w}\in \mathcal{W}\cup \{\mathbf{w}_{K+1}\}} \lVert \mathbf{q}_{t}-\mathbf{w}\rVert_{2}$ and $k^* \neq K+1$. If the query is closest to the non-yet converged 
novel representative, then it is forwarded to CS for execution. 
However, since the novel representative $\mathbf{w}_{K+1}$ is not converged, we also consider the distance from its rival (second closest) converged representative as a backup. The rival representative can provide assistance and answer the query locally instead of forwarding it to the CS if it is close enough to include $\mathbf{q}_{t}$ in the range around 
its variance $\sigma^2$.\footnote{This is associated with the vigilance parameter in Adaptive Resonance Theory dealing with the bias-plasticity dilemma.} An example is shown at Figure \ref{fig:overview-adm}, queries $\mathbf{q}_t$ and $\mathbf{q}_{t+1}$ both have $\mathbf{w}_{k+1}$ as the closest representative. However, only $\mathbf{q}_{t+1}$ will be forwarded as $\mathbf{q}_{t}$ is within the radius of $\mathbf{w}_3$. The forwarding selectivity mechanism is also evident in Algorithm \ref{process:EG} (lines 5-10). Based on the centroid theorem of convergence in vector quantization, i.e., the converged $\mathbf{w}_{k}$ is the expected query (centroid) of those queries having $\mathbf{w}_{k}$ as their closest representative, we exploit the variance $\sigma^{2} = \frac{1}{|\mathbb{Q}_{k}|}\sum_{\mathbf{q} \in \mathbb{Q}_{k}}\lVert\mathbf{q}-\mathbf{w}_{k} \rVert^{2}_{2}$ for activating the forwarding rule. The rule is based on the rival centroid and is fired if $\lVert \mathbf{q}_{t} - \mathbf{w}_{k} \rVert_{2} > \lambda \sigma$ for any scalar $\lambda > 0$ given that the query is closest to the non-converged $\mathbf{w}_{K+1}$. 
The probability of forwarding incoming queries 
from the AD to CS for execution, given that they cannot be reliably answered by the local model $\hat{g}_k$ is upper bounded 
as provided in Theorem \ref{theorem:upperbound}.
The query is  executed if inevitably 
the rival representative cannot be used for
prediction since $\lVert \mathbf{q}_{t} - \mathbf{w}_{K+1} \rVert_{2} > \lambda \sigma$. The value of $\lambda$ is adopted from 
the scaling factor of $h$, i.e., $3 \leq \lambda \leq 5$.
\begin{figure}
\begin{center}
\includegraphics[height=4cm,width=5cm,keepaspectratio]{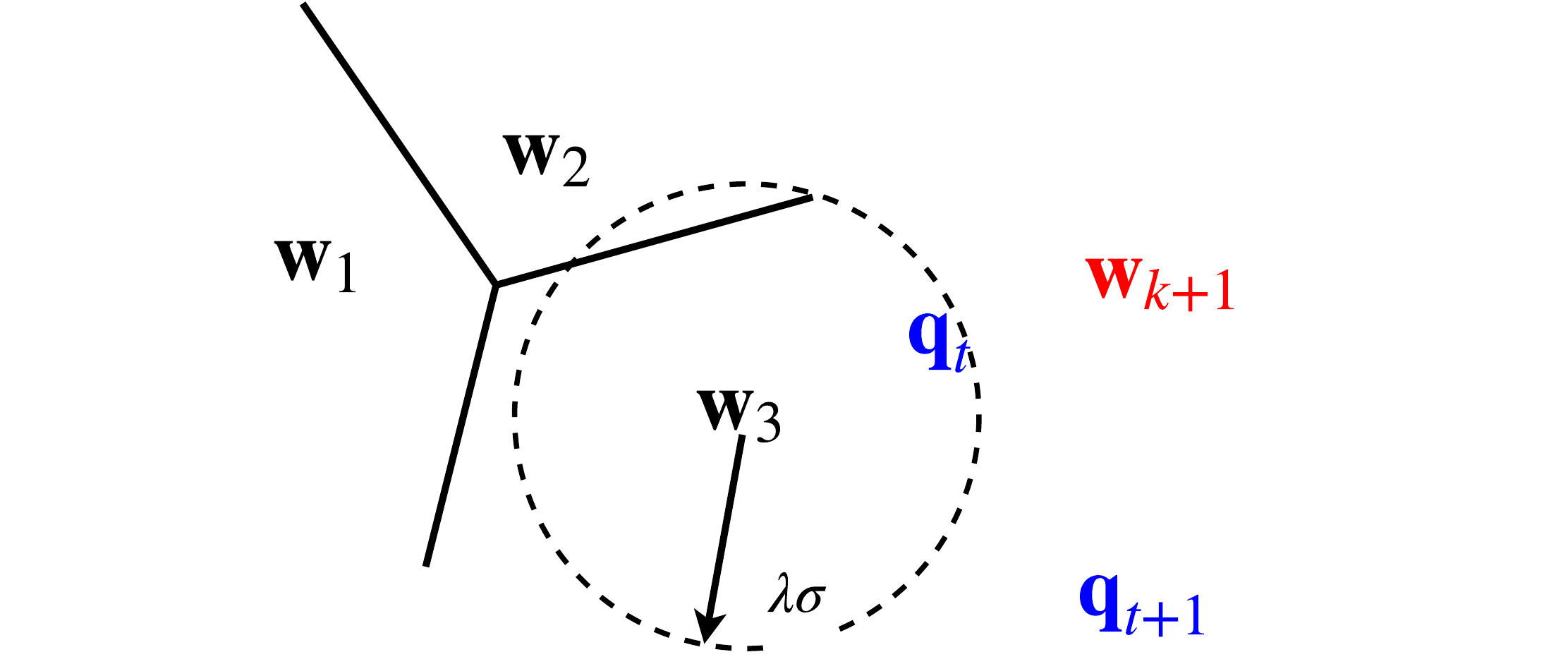}
\caption{Demonstration of the \textit{forwarding rule} by the ADM. Both queries (in blue) have $\mathbf{w}_{k+1}$ (the non-converged) as the closest representative. Only one of them is send to the CS}
\label{fig:overview-adm}
\end{center}
\end{figure}
\begin{theorem}
Given a random query $\mathbf{q}$ whose distance from its rival (second closest) representative $\mathbf{w}_{k}$ is greater than $\lambda \sigma$, the upper bound of the forwarding probability for query execution is $O(\frac{1}{\lambda^{2}})$.
\label{theorem:upperbound}
\end{theorem}
\begin{proof}
Let the query $\mathbf{q}$ being projected to its closest representative $\mathbf{w}_{K+1}$, which is not yet converged and let its second closest be the converged $\mathbf{w}_{k}$. The representative  $\mathbf{w}_{k}$ corresponds to the mean vector of those queries belonging in the cluster $\mathcal{C}_{k}$. In order to the query $\mathbf{q}$ to be forwarded to the CS for execution it means that the $\mathbf{w}_{k}$ should not be the mean vector for the incoming query $\mathbf{q}$. This is indicated if the distance $\lVert \mathbf{q} - \mathbf{w}_{k} \rVert_{2}$ is greater than a proportion of the norm of the variance $\sigma$ of the cluster $\mathcal{C}_{k}$ by a factor $\lambda > 1$. Hence, the query is sent from the AD to the CS if at least this distance is greater than $\lambda \sigma$, which is stochasitcally bounded by the factor $1/\lambda^{2}$ based on Chebyshev's inequality $P(\lVert \mathbf{q}-\mathbf{w}_{k} \rVert_{2}) \geq \lambda \sigma \leq \frac{1}{\lambda^{2}}$.
\end{proof}

%\begin{proof}
%Proof is omitted.
%\end{proof}

% The process in EN $i$ is shown in Algorithm \ref{process:EN}....

% \begin{algorithm}
% \SetKwInOut{Input}{Input}
% \SetKwInOut{Output}{Output}
% \caption{Process in EN $i$}
%  \label{process:EN}
 
%  \While{\texttt{TRUE}}{
%  receive query $\mathbf{q}_{t}$ \;
%  \eIf{\texttt{MODE} = \texttt{BUFFERING}}
%  {
% 	\If{\texttt{ADAPTATION-INIT}}
%     {send model $\mathcal{M}_{i}$ to EG \;
%       \texttt{ADAPTATION-INIT} $\leftarrow$ \texttt{FALSE}\;}
%       $\mathbf{w}_{k^{*}} = \arg \min \lVert \mathbf{q}_{t}-\mathbf{w}_{m}\rVert$ /*closest*/\;
%       $\mathbf{w}_{k} = \arg \min \lVert \mathbf{q}_{t}-\mathbf{w}_{m}\rVert$, $m \neq k^{*}$ /*rival*/\;
% 	\eIf{$\lVert \mathbf{q}_{t} - \mathbf{w}_{k}\rVert > \lambda \sigma$ and $\mathbf{w}_{K+1} = \mathbf{w}_{k^{*}}$}
%     {
%     send query $\mathbf{q}_{t}$ to EG\;
%     }
%     {$\hat{y} = \hat{g}(\mathbf{q}_{t})$ /*prediction*/\;}
%  }
%  { $\hat{y} = \hat{g}(\mathbf{q}_{t})$ /*prediction*/ \;
%   probe error $\hat{u}(\mathbf{q}_{t})$ and resources $r_{t}$ \;
%   adjust $h_{t} = f(\hat{u}(\mathbf{q}_{t}),r_{t})$ /*fuzzy inference*/\;
%   $G_{t} = \{G_{t-1} + \hat{u}_{t}\}^{+}$ /*concept-drift decision*/ \;
%   \If{$G_{t} > h_{t}$}
%   {
%   \texttt{MODE} $\leftarrow$ \texttt{EXECUTION} /*transit to execution phase*/\;
%   \texttt{ADAPTATION-INIT} $\leftarrow$ \texttt{TRUE} /*initiate adaptation*/\;
%   reset $(G_{t},U_{t}) = (0,0)$ \;
%   }%inner if
%  } %outer if
%  } %while
% \end{algorithm}
\subsection{Taking advantage of Affiliates}
In the ADM, we take advantage of the tuning process taking place in the CS. We exploit queries coming from the original AD which triggered the CDM and other queries coming from different ADs also in buffering mode due to some other \textit{independent} triggered CDMs. 
We call these potential ADs, \textit{affiliates} belonging to set $\mathcal{A}$,
since their executed queries and actual answers are used for tuning the stale models. 
Let $n$ ADs be connected to CS and assume that each 
AD system $j$ referring to model $\mathcal{M}_{j}$ enters its 
buffering mode independently of the others with 
entry probability $\beta_{j} = P(G_{j} > h) \in (0,1)$. 
Then, the probability of an AD (being in buffering mode) 
to \textit{meet} at least one affiliate $j$ in CS is $(1-\prod_{j=1}^{n-1}(1-\beta_{j}))$. 
The expected number of affiliates is then approximated by 
$\mathbb{E}[|\mathcal{A}|] \approx \beta(n-1)$ under the assumption that 
the entry probabilities are almost the same $\beta_{j} = \beta, \forall j$. This expectation will be used for studying the knowledge expansion in terms of novel query patterns being delivered to an AD via our reciprocity adaptation mechanism.

\begin{algorithm}
\SetKwInOut{Input}{Input}
\SetKwInOut{Output}{Output}
\caption{Reciprocity-based Model Adaptation in CS}
 \label{process:EG}
   CS receives (copy) of $\mathcal{M}$ and $\mathcal{W}$ from AD\;
   Set buffer $\mathcal{Q} = \emptyset$, affiliate buffer $\mathcal{Q}_{\mathcal{A}} = \emptyset$\;
   \While{\texttt{MODE} = \texttt{BUFFERING}}{
   Prediction Component receives query $\mathbf{q}_{t}$ from AD\;
   $\mathbf{w}_{k^{*}} = \arg \min_{\mathbf{w}\in \mathcal{W}\cup \{\mathbf{w}_{K+1}\}} \lVert \mathbf{q}_{t}-\mathbf{w}\rVert_{2}$ /*closest*/\;
   $\mathbf{w}_{k} = \arg \min_{\mathbf{w}\in \mathcal{W}\cup \{\mathbf{w}_{K+1}\}-\{\mathbf{w}_{k^*}\}} \lVert \mathbf{q}_{t}-\mathbf{w}\rVert_{2}$, /*rival*/\;
   \eIf{$\lVert \mathbf{q}_{t} - \mathbf{w}_{k}\rVert_{2} > \lambda \sigma$ and $\mathbf{w}_{K+1} = \mathbf{w}_{k^{*}}$}
   {
     send query $\mathbf{q}_{t}$ to CS for execution\;
   $\mathcal{Q} = \mathcal{Q} \cup \{(\mathbf{q}_{t},y_{t})\}$ /*actual query-answer pair*/ \;
   adapt prototype $\mathbf{w}_{K+1}$ \;
    }
    {$\hat{y} = \hat{g}(\mathbf{q}_{t})$ /*prediction*/\;}
   \For{affiliate $j \in \mathcal{A}$ }
   {receive affiliate query $\mathbf{q}_{j}$\;
    \If{$\lVert \mathbf{q}_{j} - \mathbf{w}_{k}\rVert_{2} \leq \lambda \sigma$} 
    {
       $\mathcal{Q}_{\mathcal{A}} = \mathcal{Q}_{\mathcal{A}} \cup \{(\mathbf{q}_{j},y_{j})\}$ /*affiliate pair*/ \;
    %adapt affiliate prototype $\mathbf{w}_{j}$\;
    }
   }%affiliate
  update learning rate $\gamma$\;
  \If{convergence w.r.t. $c$}
   {
   train new model $\hat{g}_{K+1}$ using $\mathcal{Q}$\; %and $\{\hat{g}_{j}\}$ using $\mathcal{Q},\mathcal{Q}_{j}, \forall j \in \mathcal{A}$\; 
   }%convergence
   $\mathcal{M} = \mathcal{M} \cup \{\hat{g}_{K+1},\hat{g}_{j}\}, \mathcal{W} = \mathcal{W} \cup \{\mathbf{w}_{K+1},\mathbf{w}_{j}\}, \forall j \in \mathcal{A}$\;
   $\mathcal{M}$ and $\mathcal{W}$ are sent to AD from CS\;
   set \texttt{MODE} = \texttt{PREDICTION}
   }
\end{algorithm}

\subsection{Model Adaptation \& Reciprocity}
In the CS, when a query is selectively forwarded 
from AD, the process of model adaptation has as follows: 
for adapting to new query patterns, 
we rely on the principle of \textit{explicit partitioning} 
\cite{tsymbal2004problem,ditzler2015learning}, 
as a natural extension of our strategy using an ensemble of local ML models. To adjust to new query patterns, 
we train a new model $\hat{g}_{K+1}$ using 
 executed queries and their answers in CS. 
This is the optimal strategy for expanding the 
current $\mathcal{M}$ as other methods might lead to \textit{catastrophic forgetting} \cite{gepperth2016incremental}. 
Indicatively, such methods adopt strategies to adapt the 
current model by adjusting to new patterns 
whilst forgetting the old ones. In our context, 
this is not applicable since analysts have the flexibility to issue queries either conforming to the \textit{old} patterns or to the new ones, depending on the analytics process.  

The adaptation process is performed with parameters: the 
$K$ query prototypes $\mathcal{W}$ and their associated ML models $\mathcal{M}$ as shown in Algorithm \ref{process:EG}. Recall that the analyst's device has cached models $\mathcal{M}$ and the 
DC/CS adapts the received parameters by learning the new underlying query patterns and based on these
trains the new ML model. Let the queries series $\{\mathbf{q}_{1}, \mathbf{q}_{2}, \ldots \}$ coming from the AD to CS based on selective forwarding. This means that most likely a query $\mathbf{q}_{t}$ conforms to new query patterns thus sent to CS for execution. Once query $\mathbf{q}_{1}$ is executed and its actual answer $y_{1}$ is obtained, it is then considered as a new (initial) representative $\mathbf{w}_{K+1}$ for $\mathcal{M}_{K+1}$. The pairs $(\mathbf{q}_{t},y_{t})$ are then used to incrementally update $\mathbf{w}_{K+1}$ and then buffered in $\mathcal{Q}$, which will be the training set for $\hat{g}_{K+1}$ (lines 4-13).  The adaptation of $\mathbf{w}_{K+1}$ to follow the new query pattern is achieved by Stochastic Gradient Descent (SGD) \cite{bottou2012stochastic}, which is widely used in statistical learning for training in an on-line manner considering one training example (query-answer) at a time. We focus on the convergence of the query distribution by moving the new query representative towards the estimated median of the queries in $\mathcal{Q}$ and not the corresponding centroid. 
This is introduced so that the new representative converges to a robust statistic, free of outliers and more reliable than the centroid (mean vector). The convergence to the median denotes with high reliability \textit{convergence to the distribution}, which is what we desire for model convergence. In this case, we provide the adaptation rule of the new query representative to converge to the median of the forwarded queries, as provided in Theorem \ref{theorem:median}.

\begin{theorem}
The novel representative $\mathbf{w}_{K+1}$ converges to the median vector of the queries executed in the DC w.r.t. update rule $\Delta \mathbf{w}_{K+1} \propto \gamma sgn(\mathbf{q}-\mathbf{w}_{K+1})$, $\gamma \in  (0,1)$; $sgn(\cdot)$ is the signum function. 
\label{theorem:median}
\end{theorem}
\begin{proof}
Each dimension $i$ of the median vector $\mathbf{m}$ of the 
queries $\mathbf{q}$ in a sub-space satisfies: 
$P(q_{i} \geq m_{i}) = P(q_{i} \leq m_{i}) = \frac{1}{2}$. 
Suppose that the new representative $\mathbf{w}_{K+1}$ has 
reached equilibrium, i.e., $\Delta \mathbf{w}_{K+1} = \mathbf{0}$ holds with probability 1.
By taking the expectations of both sides 
of the update rule $\mathbb{E}[\Delta \mathbf{w}_{K+1}] = \alpha \mathbb{E}[sgn(\mathbf{q}-\mathbf{w}_{K+1})] = \mathbf{0}$ and focusing on each dimension $i$, we obtain that: $\int sgn(q_{i}-w_{K+1,i})p(q_{i})dq_{i} = P(q_{i} \geq w_{K+1,i})\int p(q_{i})dq_{i} - P(q_{i} < w_{K+1,i})\int p(q_{i})dq_{i} = 2P(q_{i} \geq w_{K+1,i}) - 1$. Since $\mathbb{E}[\Delta w_{K+1,i}]=0$ is constant, then $P(q_{i} \geq w_{K+1,i}) =  \frac{1}{2}$, which denotes that $w_{K+1,i}$ converges to the median of $q_{i}$, $\forall i$.
\end{proof} 

Given the $j$-th incoming query $\mathbf{q}_{j}$ issued by analysts in the $j$-th affiliate AD, the CS assesses the selectivity forwarding criterion: $\lVert \mathbf{q}_{j} - \mathbf{w}_{k}\rVert_{2} \leq \lambda \sigma$. If it holds, these affiliate queries are exploited for 
expanding the query patterns (lines 14-18). 
That is, CS buffers the pairs $(\mathbf{q}_{j},y_{j})$ 
in affiliate set $\mathcal{Q}_{\mathcal{A}}$ which will be used later for 
training new ML models $\hat{g}_{j}$ enriching the 
predictability variety of $\mathcal{M}$. 
In this case, we obtain the affiliate new representative $\mathbf{w}_{j}$ generated by query patterns coming from the affiliate AD $j \in \mathcal{A}$. Similarly to the new $\mathbf{w}_{K+1}$, affiliate $\mathbf{w}_{j}$ is incrementally adapted through SGD with 
the aim to converge to the corresponding 
median of the affiliate queries in $\mathcal{Q}_{ \mathcal{A}}$. For the new $\mathbf{w}_{K+1}$ and the possibly affiliate $\mathbf{w}_{j}$, the median convergence rule involves the learning rate: 
\begin{eqnarray}
\gamma & = & \frac{1}{1+|\mathcal{Q}| +|\mathcal{Q}_{\mathcal{A}}|},
\end{eqnarray}
which decreases as more queries are appended to $\mathcal{Q}$ and $\mathcal{Q}_{\mathcal{A}}$; the higher the number of affiliate query representatives, the faster the convergence to the median. 
This demonstrates the exploitation of affiliates to the adaptation 
of $\mathcal{M}$. 

The convergence of the representatives is checked by the subsequent adjustments in positions that $\mathbf{w}_{K+1}$ makes. If that change is lower than a threshold $c$ then convergence has been achieved. After the convergence of the query representative and affiliates (if any), the CS trains the new models $\hat{g_{K+1}}$ and $\{\hat{g}_{j}\}$, using the $\mathcal{Q}$ and $\mathcal{Q}_{\mathcal{A}}$, respectively. The new ML models and new representatives are then delivered to AD (lines 20-26). 
Evidently, set $\mathcal{W}$ is now expanded with one more 
representative $\mathbf{w}_{K+1}$ and on average $(n-1)\beta$ affiliate representatives along with their regression models. The adapted and updated $\mathcal{M}$ is expected to have $K + 1 + (n-1)\beta$ query representatives and ML models at the end of the buffering phase. 

\subsection{Convergence to an Offline Mode}
When the system transits from the buffering 
to prediction mode, the enhancement of $\mathcal{M}$ and 
$\mathcal{W}$ gradually decreases the probability to 
enter the buffering mode in the future, in light of 
learning the query patterns not only derived from 
the analysts interacting with the CS/DC, but also the patterns from other analysts in other ADs. This indicates
that the gradually expanding sets reflect the analysts' 
way of exploring and analyzing data among data centers. 
Because of this expansion, the transition probability from prediction to buffering mode gradually decreases saving computational and communication resources at the network and CS. 
The expected ratio of new models in an AD transiting from 
the $m$-th buffering mode with $K_{m}$ representatives to 
the $(m+1)$-th prediction mode is: $1 + \frac{1+\beta(n-1)}{K_{m}}$, with 
$\mathbb{E}[\Delta K_{m}] = \mathbb{E}[K_{m+1}-K_{m}] = 1 + \beta(n-1)$.
Such ratio increases to unity after certain prediction-buffering transitions denoting that \textit{all} query sub-spaces are known: 
\begin{eqnarray}
\lim_{m \to \infty}(1+\frac{1+\beta(n-1)}{m}) = 1.
\end{eqnarray}
An AD model then learns \textit{all} possible query sub-spaces 
via its analysts and affiliate models with rate $O(\frac{1}{m})$. The entry probability to buffering mode decreases with the same rate, thus, reducing the CS execution overhead and communication load by 
transiting the AD to 'offline' mode. This is the advantage of the query-driven analytics over dynamic workloads with 
the expected query execution rate in CS being bounded:
\begin{theorem}
The expected query execution rate in the CS is bounded 
by $O(\frac{1}{\lambda^{2}}(2-(1-\beta)^{n-1}))$.
\label{theorem:execution}
\end{theorem}
\begin{proof}
Consider that at a certain time instance, there are $n-1$ affiliate ADs which enter in their buffering mode with entry probability $\beta$. Then, the
probability of existing at least one affiliate of an AD in the CS in the buffering mode is $1-(1-\beta)^{n-1}$. Given that a query is forwarded with upper probability $O(\frac{1}{\lambda^{2}})$ for those ADs being in the buffering mode, then, the expected number of queries being executed in the CS from an AD and its affiliates within any arbitrary time interval $T$ is $T(\frac{1}{\lambda^{2}})(1 + 1-(1-\beta)^{n-1})$.
\end{proof}
\section{Evaluation Results}
\label{sec:evaluation}
The main questions we are striving to answer in our evaluation are the following :
\begin{enumerate}
    \item How accurate are the given predictions for a variety of aggregate queries ?
    \item Is there a single ML model that can be used for this purpose ?
    \item What are the effects of predicates and data dimensionality (number of columns) on estimating the results of aggregate queries ?
    \item How light-weight and efficient are the models and can they be stored on ADs for efficient execution
    \item How effective are the CDM/ADM mechanisms and what is the effect of continuously learning and adapting to new queries ?
\end{enumerate}{}
\subsection{Implementation \& Experimental Environment}
To implement our algorithms we used \texttt{scikit-learn} , \texttt{XGBoost}\cite{chen2016xgboost} and an implementation of the \texttt{Growing-Networks} algorithm \cite{marsland2002self}. We performed our experiments on a desktop machine with a Intel(R) Core(TM) i7-6700 CPU @ 3.40GHz and 16GB RAM. For the real datasets, the GrowingNetworks algorithm was used for clustering mainly because of its invariance to selecting a pre-defined number of clusters and it's ability to naturally grow (as required by our adaptability mechanisms). XGBoost was used as the supervised learning algorithm because of its superior accuracy to other algorithms we tested and shown as part of our evaluation. 

\textbf{Real datasets}: We use the \textbf{Crimes} dataset from \cite{crimesdata} and the \textbf{Sensors} dataset from \cite{sensors}. The Crimes dataset contains $|\mathcal{R}_1| = 6.6\cdot{10^6} $ and the Sensors dataset $|\mathcal{R}_2| = 2.3 \cdot{10^6}$ data vectors. We created synthetic query workloads over these datasets as real query workloads do not exist for this purpose as also attested by \cite{wasay2017data}. 
For Crimes, we generated predicates restricting the spatial dimension and for Sensors the temporal dimension as essentially this is what analysts would be doing in exploration tasks. For the predicates in the spatial dimension we used multiple multivariate-normal distributions to simulate the existence of multiple users. 
For the temporal dimension, we used a uniform distribution. We then recorded the results of the descriptive statistics \texttt{COUNT}, \texttt{MEAN}, \texttt{SUM} over different attributes in the datasets to sufficiently make sure the workload is randomized. 

\textbf{Synthetic datasets}: We also generated synthetic datasets and workloads to stress test our system. We generated a varying number of predicates and attributes to see how it would affect a state-of-the-art model chosen by an initial study comparing different models under different aggregates. This helped us understand the implications of our chosen representation and identified under what conditions the accuracy deteriorates.

\subsection{Predictability}
\begin{figure}[htbp]
\begin{center}
\includegraphics[height=4cm,width=6cm, keepaspectratio]{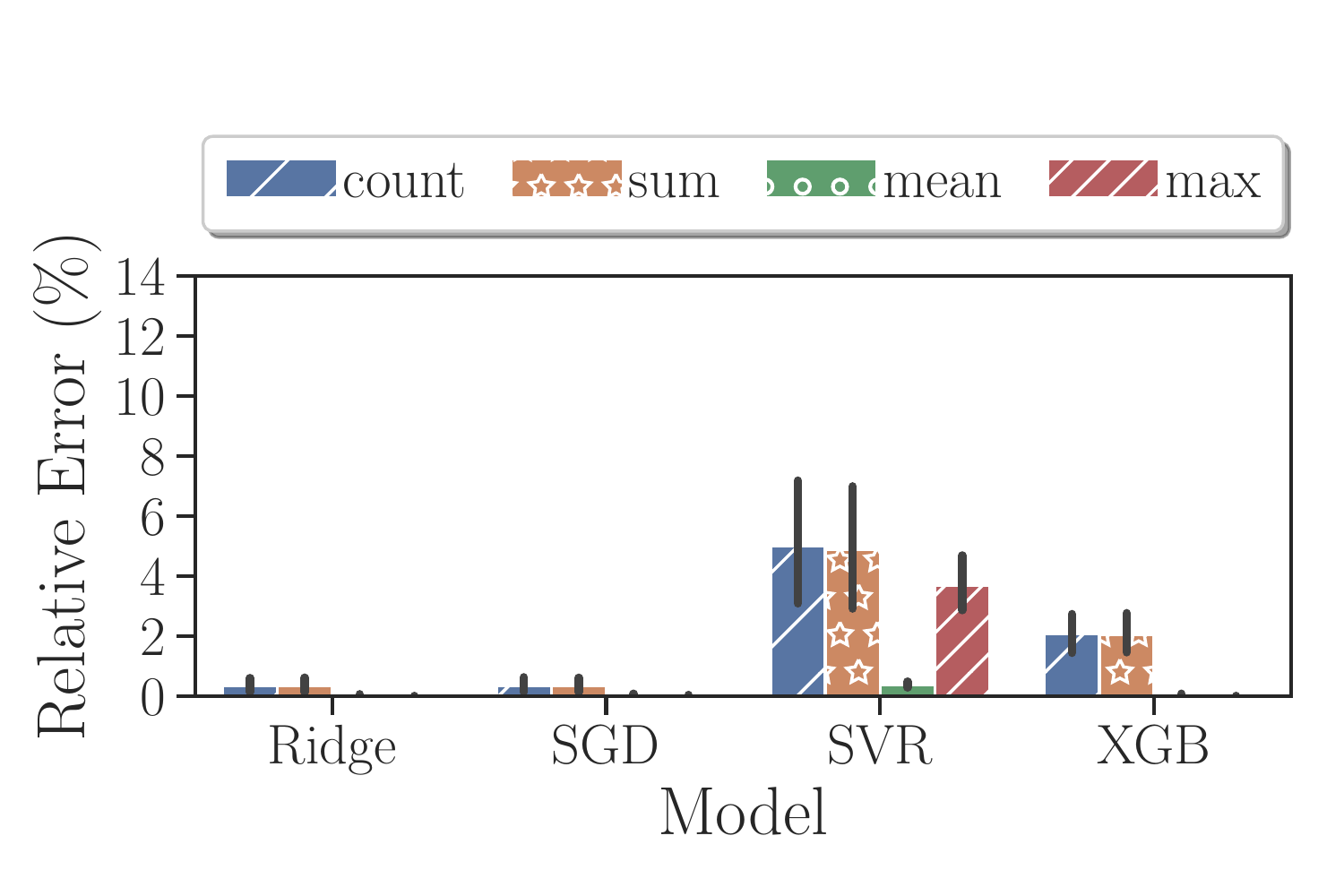}
\caption{Relative prediction error for descriptive statistics.}
\label{fig:accuracy-aggsXmodels}
\end{center}
\vspace{-0.25cm}
\end{figure}

We measured the prediction accuracy of our system using both synthetic and real datasets. We examine the median relative error, unless stated otherwise.
We first examine the \textit{predictability} of different 
descriptive statistics using a variety of supervised ML models. The experiments were conducted using our synthetic workloads to test accuracy across a varying number of attributes/columns and predicates. The models are trained using $80\%$ of the total queries ($=10^4$) and tested against $20\%$. 
Where: \texttt{Ridge} is for Regularized Linear Regression model, \texttt{SGD} is a linear regression model trained on-line using SGD, \texttt{SVR} is for Support Vector Regression with 
RBF kernel, \texttt{XGB} is the XGBoost model. 
None of the models was hyper-tuned to provide the 
purest accuracy as we desired to test whether different statistics could be better estimated by different models, 
thus indicating the need to choose optimal models at training. Figure \ref{fig:accuracy-aggsXmodels} shows the results of this experiment. The main takeaways are as follows. First, query-driven learning consistently produces low relative errors across all statistics. Second, it is largely insensitive to underlying ML models; given our proposed representation, all models are able to predict the given statistics with small error, well below $10\%$. 
Third, there is some high variation across the reported error as all of the workloads with varying predicates and columns were used.
Finally, all statistics can be adequately predicted by a single algorithm, that being \texttt{XGB}. 
The latter represents an advantage for this work, as the system can be optimized for storing such models and can be designed around this single class of model instead of trying to accommodate a variety of them, each with its own restrictions.

\begin{figure}[!htbp]
\begin{center}
\begin{tabular}{cc}
\includegraphics[height=4cm,width=4.5cm,keepaspectratio]{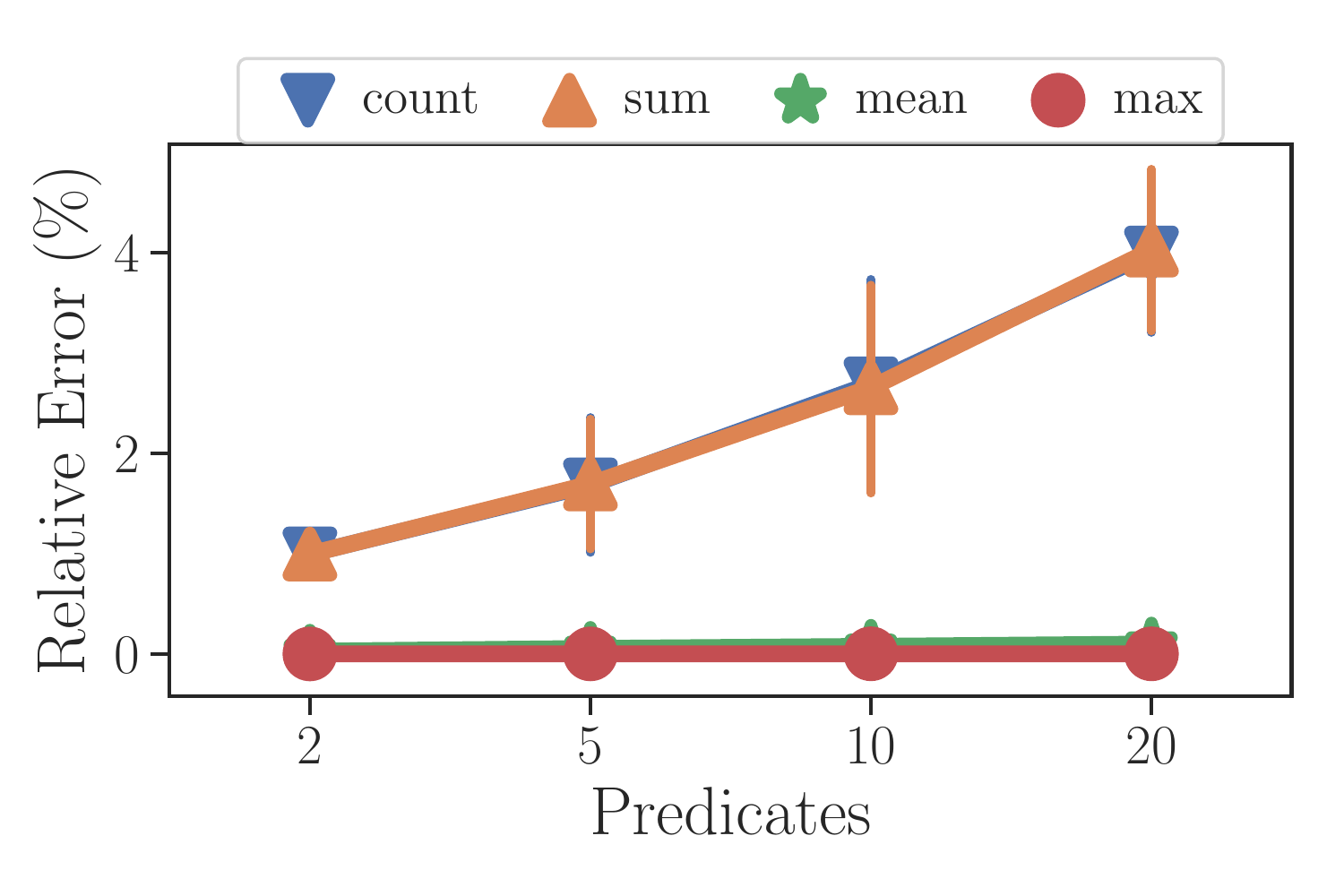}
\includegraphics[height=4cm,width=4.5cm,keepaspectratio]{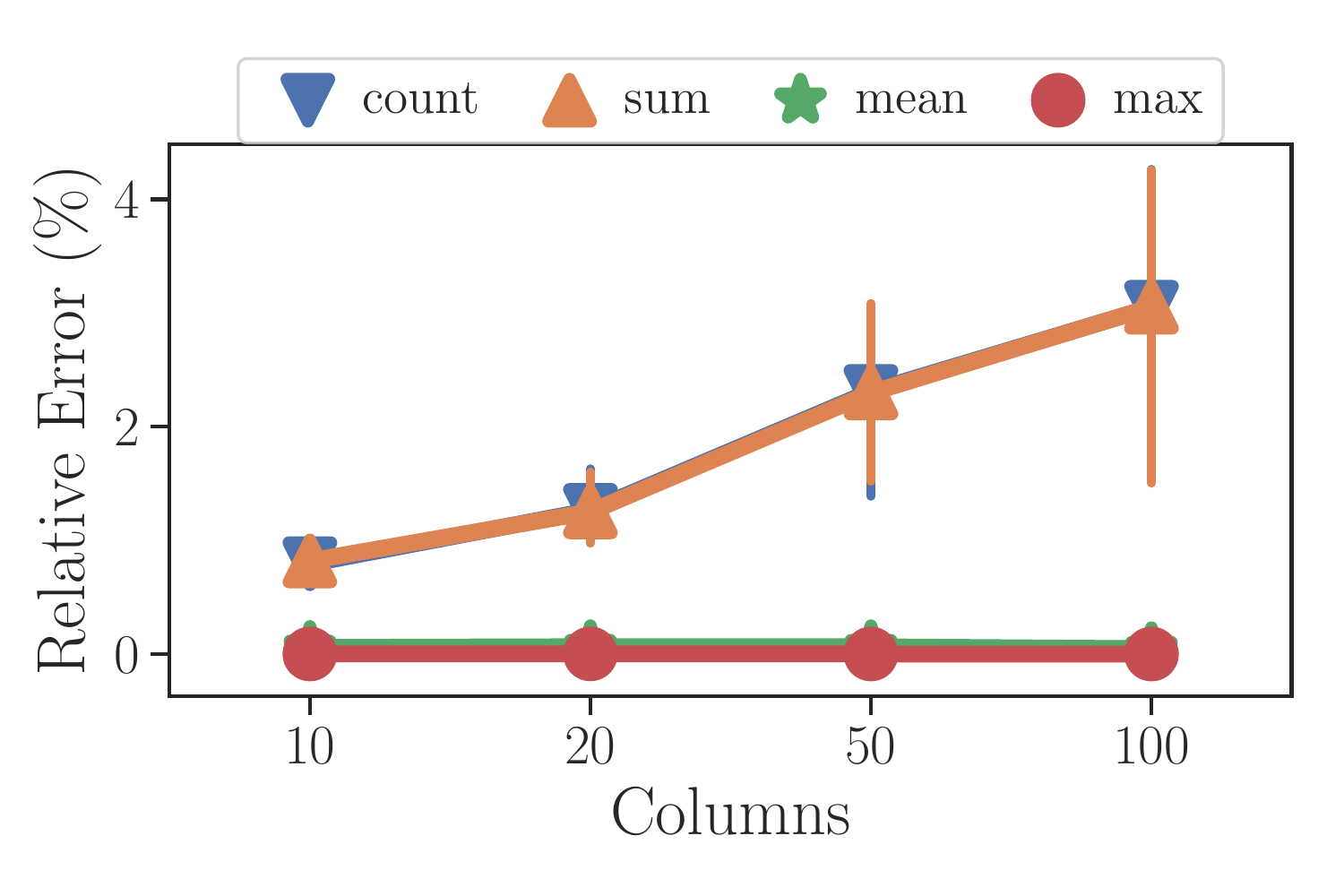}\\
\end{tabular}
\caption{Relative error of statistics vs. increasing number of (left) predicates, (right) columns.}
\label{fig:acc-vs-conditions}
\end{center}
\end{figure}

Using the most accurate model derived from our experiments (\texttt{XGB}), 
we evaluated the prediction accuracy over different statistics 
under different conditions, shown in Figure \ref{fig:acc-vs-conditions}. As expected, the increase in number of predicates used and number of columns has a negative impact on accuracy. The representation used is high dimensional as the query vector size is $\mathbb{R}^{2d}$, where $d$ is the number of columns in a dataset. Hence, for $d=100$ columns used, Figure \ref{fig:acc-vs-conditions}(right), the model is trained using $\mathbb{R}^{200}$. This is no trivial task and is stress testing our system's capabilities. In addition, the number of predicates used is the number of $(l_i,u_i)$ elements which restrict the sub-space and are sparsely populated. Again this can make the fitting process harder. However, even under these conditions, the increase in error is no more than linear w.r.t to increasing the number of predicates and columns. With statistics such as \texttt{MAX}, being less impacted by this change, as their error seems to be invariant from the beginning. 
Nonetheless, the results are reassuring: The proposed approach delivers very low errors, even when large number of predicates and columns are used in queries.  To put things in perspective the median number of columns selected in a query is around $8$ \cite{kandula2016quickr} in typical workloads and also followed by our proposed workload.

We experiment with real datasets to 
demonstrate the applicability of our system under real conditions. As evident from Figure \ref{fig:acc-vs-queries} our system provides 
estimations for descriptive statistics over 
different types of real datasets with relative error below $10\%$.
(A relative error below 10\% is the target of modern state of the art approximate-answer production systems \cite{kandula2016quickr}). We also tested these datasets usng VerdictDB \cite{park2018verdictdb}, a state-of-the-art system in Approximate Query Processing. The errors obtained varied from $1\%-14\%$ with sampling ratio of $1\%-10\%$. These results are comparable to ours and show that our system can be reliably used in parallel to such engines when local access is needed and resource consumption at CS is to be minimized.
After training the system with more than ca. 2000 queries the  relative error starts approaching its minimum value rather swiftly for both datasets (named \textbf{Crimes} and \textbf{Sensors}. This demonstrates the capabilities of the proposed learning approach  to offer high accuracy estimates for approximating analytical query answers with only a fairly small number of training queries. Note that typical industrial-strength in-production big data analytics clusters used for approximating answers to such analytical queries receive several million of queries per day \cite{kandula2016quickr}. Therefore, one can expect that a system employing our approach would receive a few thousand of training  queries in a just few tens of seconds. 

% \begin{figure}[!htbp]
% \begin{center}
% \includegraphics[height=4cm,width=6cm,keepaspectratio]{crimes-sensors-rel-error}
% \caption{Relative error for different descriptive statistics across two real datasets}
% \label{fig:accuracy-aggsXdata}
% \end{center}
% \end{figure}

\begin{figure}[!htbp]
\begin{center}
\begin{tabular}{cc}
\includegraphics[height=4cm,width=4.5cm,keepaspectratio]{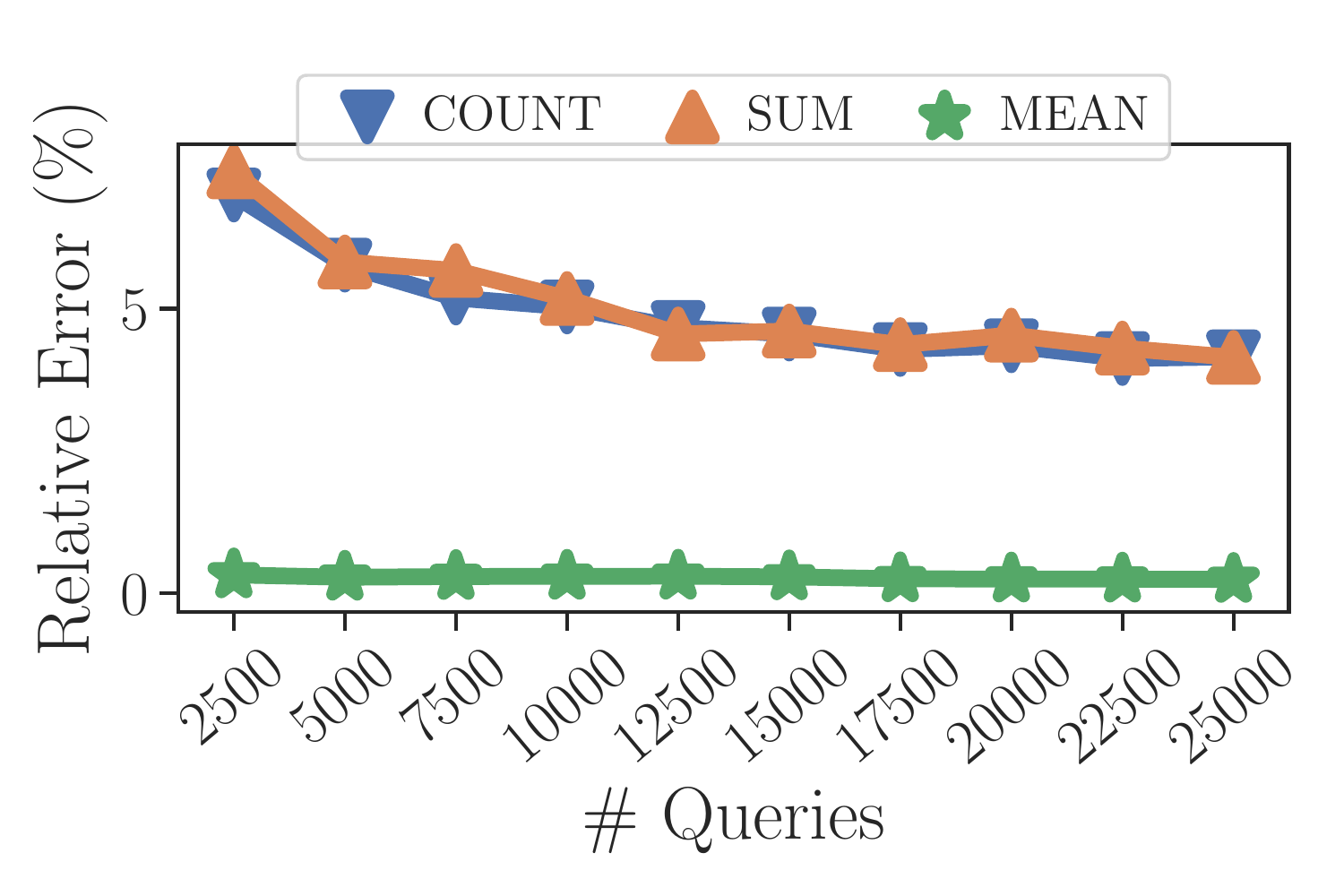}
\includegraphics[height=4cm,width=4.5cm,keepaspectratio]{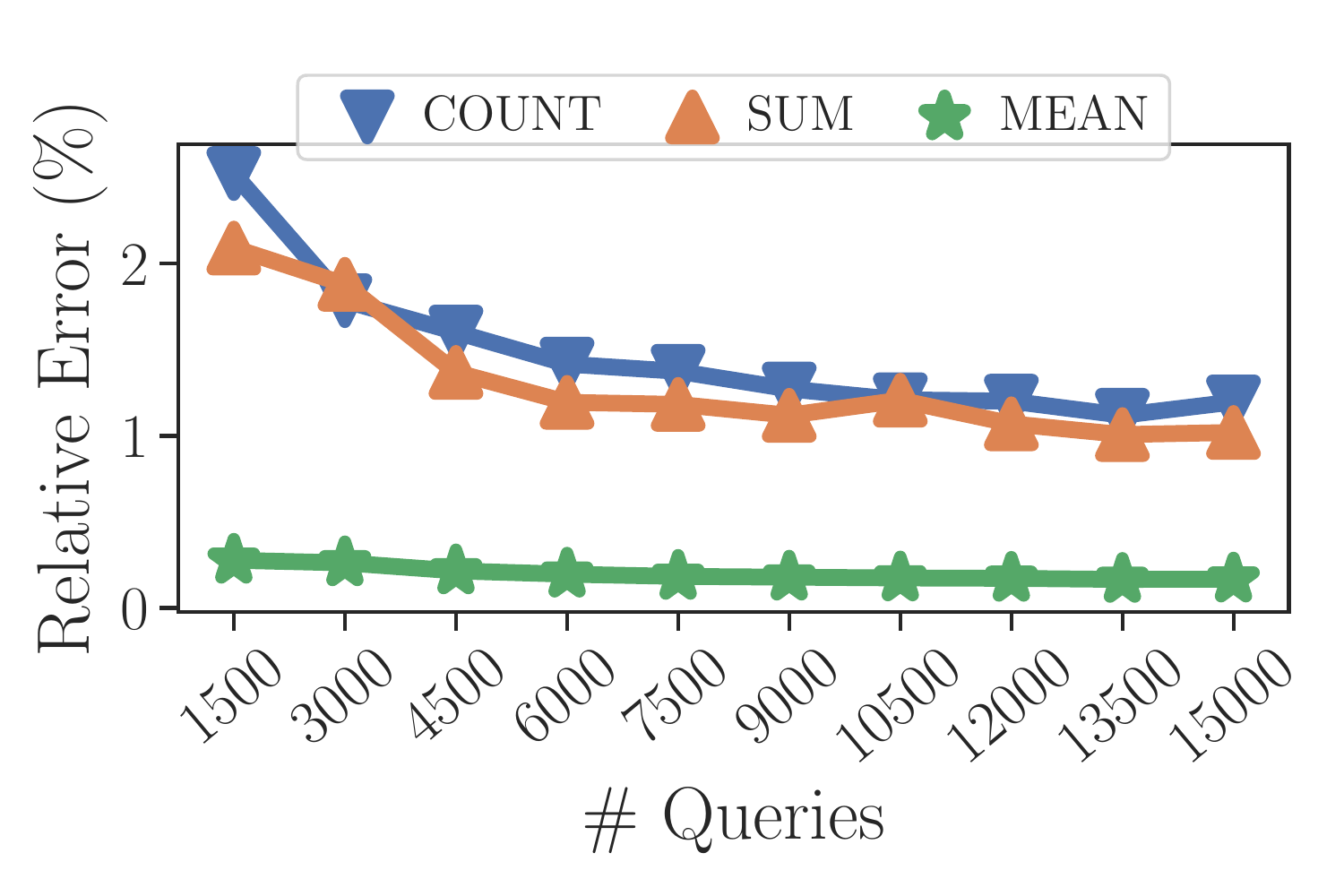}\\
\end{tabular}
\caption{Relative error vs. number of training queries; (left) Crimes (right) Sensors.}
\label{fig:acc-vs-queries}
\end{center}
\end{figure}
\subsection{Performance \& Storage}
We examine the performance and storage requirements of our system. 
This is important as our solution has to be light-weight both 
in terms of storage overhead for ADs and efficient in 
transferring models through the network. 
We examine all the above-mentioned models 
to identify the most efficient ones in training and prediction. 
A synthetic workload with $50$ predicates and $100$ columns is used for training all the models.
For \textit{Prediction Time} (PT) in Table  \ref{tab:perf_stor} 
we report on the expected prediction time and standard deviation 
of each model. As expected, \texttt{SVR} has the worst performance. 
The central takeway here is that PTs are negligible -- much less than a millisecond, thus, 
guaranteeing efficient statistics estimation irrespective of the adopted model. Even though there are multiple models trained, to account for varying query patterns, the time complexity associated is $\mathcal{O}(K)$ with $K$ usually being small. We show that our learning models based on mining query logs and learning using query-answer pairs, can go a long way in amplifying the capabilities of analytic system stacks as they can act as a 'caching' mechanism without actually storing the results of past queries but instead using models to perform answer estimation for new queries. 

For measuring the training time of individual models, we varied the number of training 
samples and examined the expected model Training Time (TT) in 
Table \ref{tab:perf_stor}. We used $11$ training samples varying in size wrt $\{4\cdot 10^2, \ldots, 4\cdot 10^5\}$. For \texttt{SVR}, we stop 
recording after $1.6\cdot 10^4$ training samples as the algorithm is no longer efficient and should be avoided. 
% As expected, the performance of all models begin to deteriorate as a there are a number of optimizations going on when trying to minimize the objective function given to them. 
Although \texttt{XGB} appears to perform the worst, 
we note that its TT is no more than $53$ seconds without using its multi-threading capabilities. 
% \begin{figure}[!h]
% \begin{center}
% \begin{tabular}{cc}
% \includegraphics[height=4cm,width=4.5cm,keepaspectratio]{perf-ave-pred-time}
% \includegraphics[height=5cm,width=4cm,keepaspectratio]{perf-training}\\
% \end{tabular}
% \caption{Performance across different models (Left) Average Prediction time (Right) Training time across models with an increasing number of tuples.}
% \label{fig:perf}
% \end{center}
% \end{figure}

We also, examine the model Size in KB shown in Table \ref{tab:perf_stor}
and observe that it refers to a 
negligent cost fulfilling our initial requirements about 
the system being light-weight to reside in the ADs memory. 
The results are for one individual model therefore the resulting 
storage cost is $K$ times the initial one. 
This could be a significant overhead, especially for \texttt{SVR} 
as $K \gg 1000$. However, the cost incurred is not preventive as the benefit of decreasing latency times, offloading queries otherwise issued to the cluster and no extra monetary cost are far greater.
\begin{table}[htbp]
    \centering
    \begin{tabular}{lrrr}
    \hline
    {} &       Size (KB)  &       TT (s) &    PT (ms)  \\
    \hline
    Ridge &     $1.45\pm0.6$ &   $0.55\pm1$ &     $0.0008\pm0.0004$ \\
    SGD   &     $1.73\pm0.6$  &   $0.44\pm0.8$ &  $0.0008\pm0.0007$ \\
    SVR   &  $1332.68\pm944$  &   - &  $0.14\pm0.078280$ \\
    XGB   &    $65.42\pm 4$ &  $52.58\pm90$ &  $0.008\pm0.005201$ \\
    \hline
    \end{tabular}
    \caption{Performance and Storage results across models}
    \label{tab:perf_stor}
    \vspace{-0.25cm}
\end{table}

% \begin{figure}[!htbp]
% \begin{center}
% \includegraphics[height=4cm,width=6cm]{per-stor}
% \caption{Storage requirements for different models}
% \label{fig:storage}
% \end{center}
% \end{figure}

\subsection{Adaptivity}

\begin{figure}[!htbp]
\begin{center}
\begin{tabular}{cc}
\includegraphics[height=4cm,width=4.5cm,keepaspectratio]{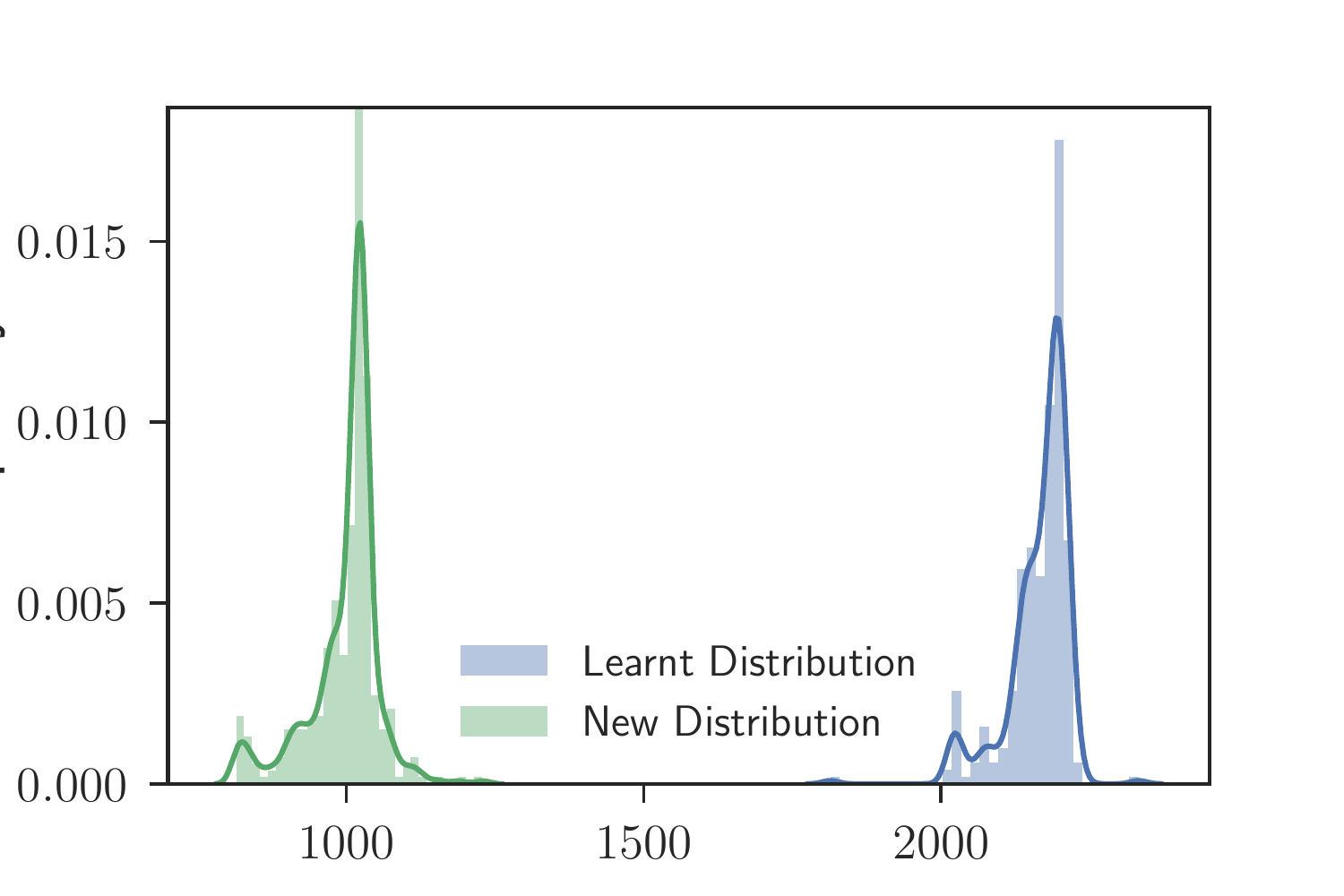}
\includegraphics[height=4cm,width=4.5cm,keepaspectratio]{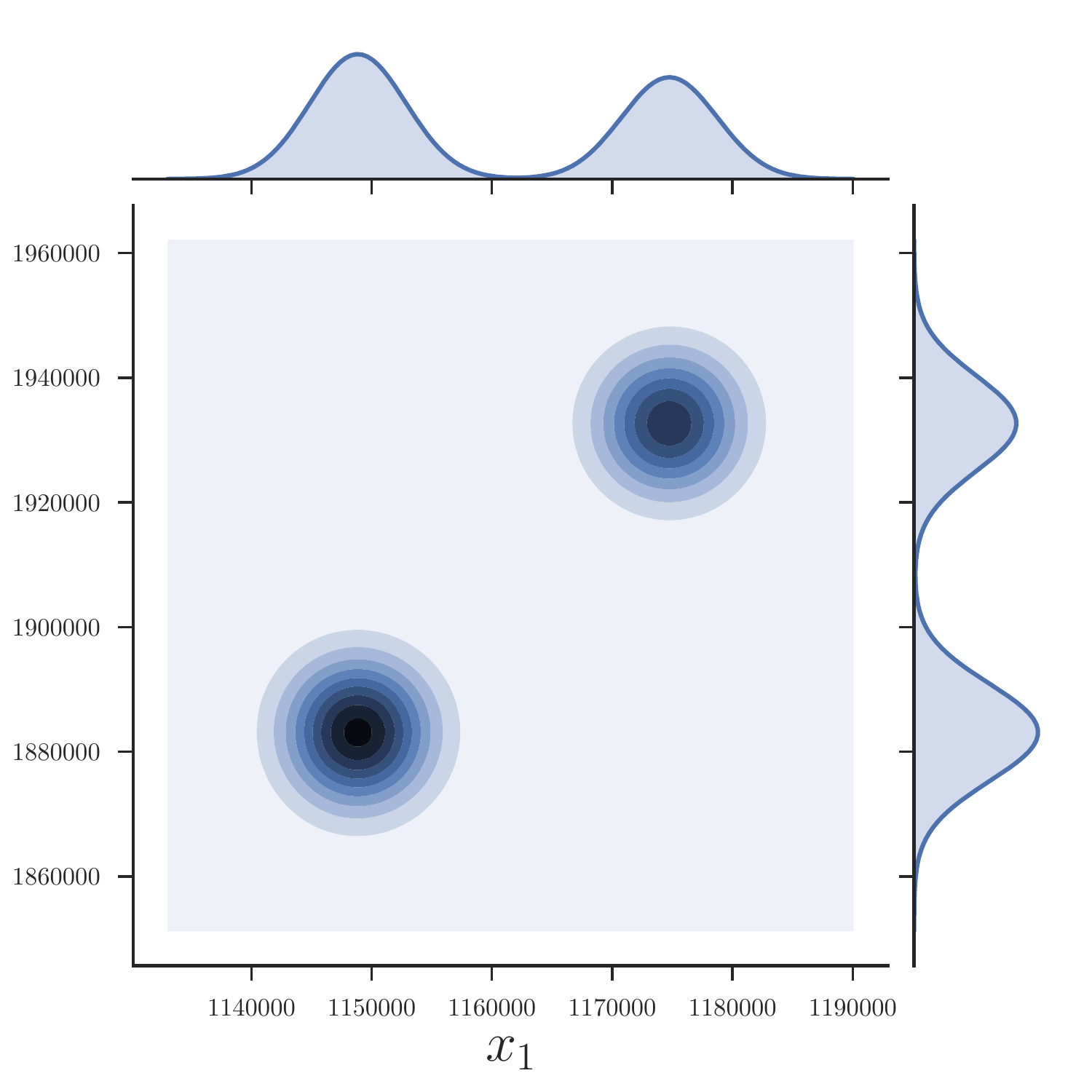}\\
\end{tabular}
\caption{(Left) Inducing concept drift 
of learned distribution of $y$; 
(right) Lower left is the initial query pattern distribution; upper right is the new one.}
\label{fig:concept-drift-experiment}
\end{center}
\end{figure}

To examine our CDM, ADM due to concept drift 
we devised the following experiment: 
consider a $\mathcal{M}_i$ 
that has already learned a particular distribution of $y$ being deployed to answer queries. 
At a particular point in time the query patterns 
might shift as shown in Figure \ref{fig:concept-drift-experiment}. 
Figure \ref{fig:concept-drift-experiment} 
shows two different query distributions. Figure \ref{fig:concept-drift-experiment}(left) are the distributions of the query answers $y$. Their respective query patterns are shown in Figure \ref{fig:concept-drift-experiment} (right) \footnote{Only two dimensions shown for visualization purposes}. Our aim is to examine whether CDM 
detects a query pattern shift from one distribution to another. If remained undetected, 
it will cause detrimental problems in accuracy 
due to different distributions of $y$. 
We first set the detection threshold $h=3\sigma_{\Tilde{u}}$ and 
convergence threshold $c=0.008$; a following sensitivity analysis 
shows the impact of these tuning parameters on ADM and CDM. 
We gradually introduce new query patterns and compare 
our system with an approach where no adaptation is deployed. 
Figure \ref{fig:adaptivity} shows the different queries being processed by our mechanism and the associated \textit{true} prediction error. 
We first measure the error of queries using the \textit{known} distribution until $t=66$. 
From that point onwards, we shift to the unknown distribution and 
evidently the error increases dramatically should no adaptation 
mechanism be employed. On the other hand, 
CDM detects that a shift has happened and 
transits the system from prediction mode to 
buffering mode until the exiting criteria are met. 
At the end, a new model is introduced which is trained using 
the new distribution as evidenced by the decreased error at $\mathcal{M}_{new}$.
\begin{figure}[!htbp]
\begin{center}
\includegraphics[height=4cm,width=6cm,keepaspectratio]{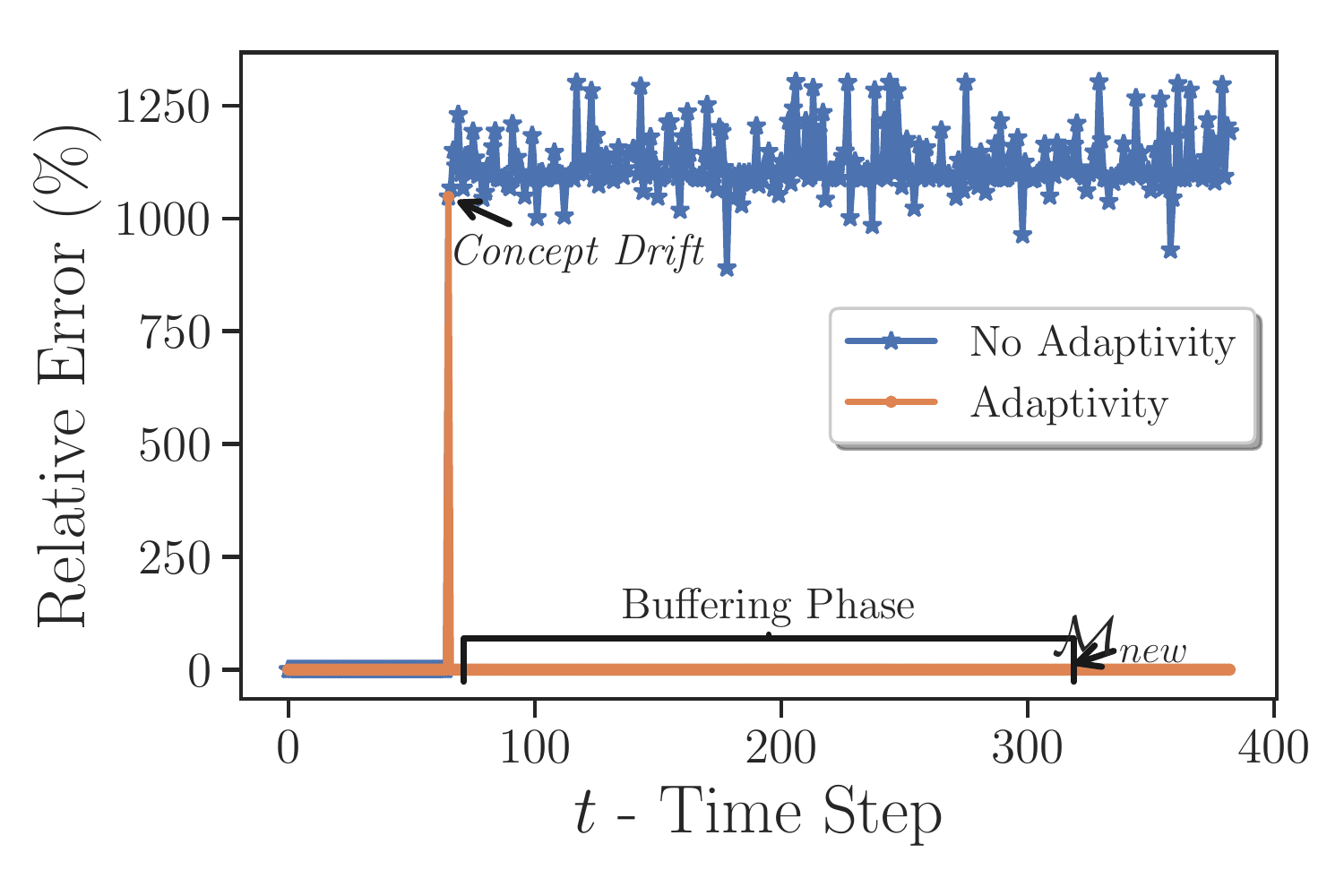}
\caption{Error with concept drift detection/adaptation.}
\label{fig:adaptivity}
\end{center}
\end{figure}

Parameters $h$ and $c$ are responsible for the ADM and CDM
with the impact of $c$ shown in Figure \ref{fig:sensitivity-c}(left). 
As we increase $c$ we allow for an early exit from buffering mode. 
An early exit means that less queries have been processed thus potentially 
the examples are not sufficient for accurately learning the distribution as witnessed in Figure \ref{fig:sensitivity-c}(left), where the relative error increases, therefore the accuracy decreases as we increase $c$. 
\begin{figure}[!htbp]
\begin{center}
\begin{tabular}{cc}
\includegraphics[height=4cm,width=4.5cm,keepaspectratio]{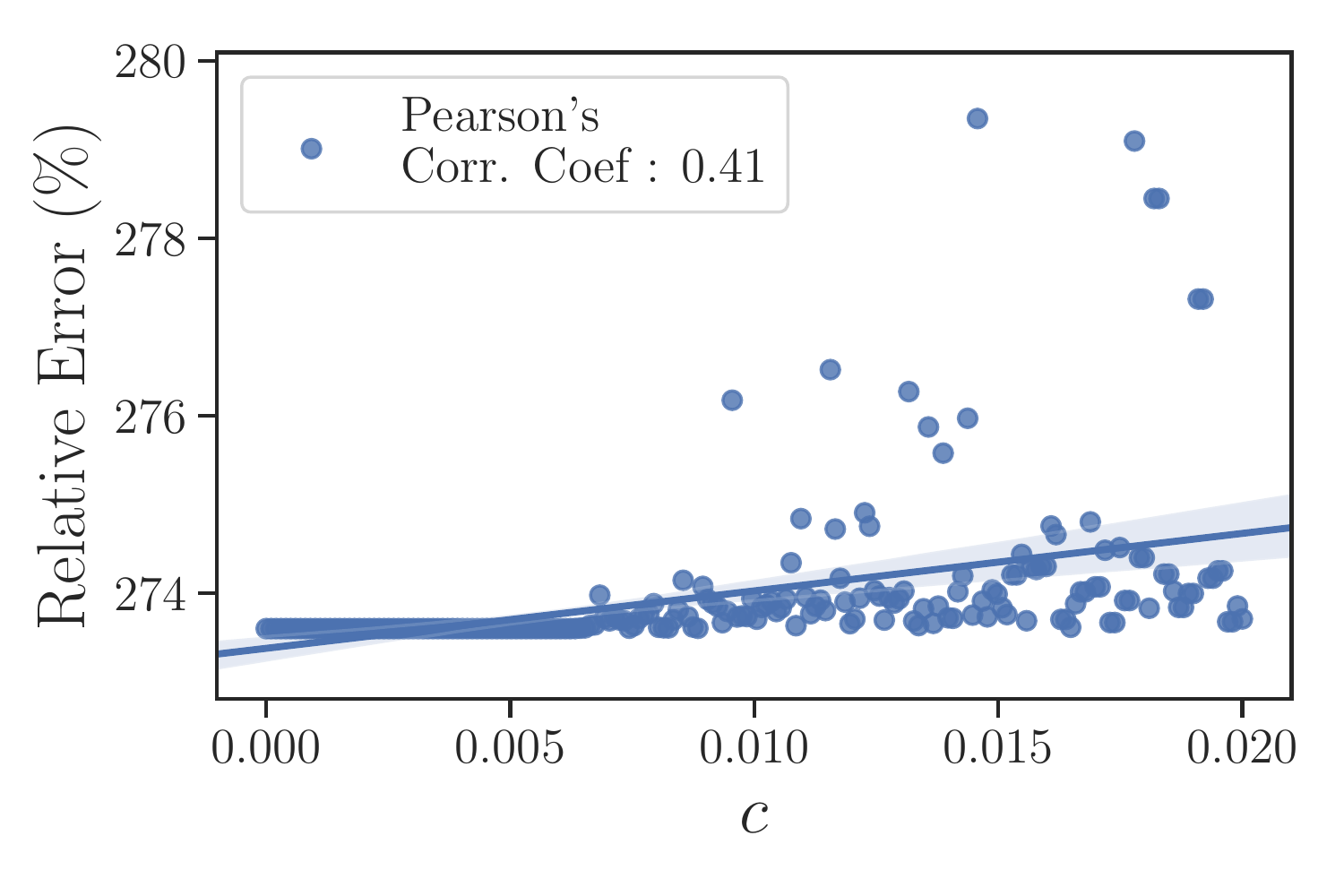}
\includegraphics[height=4cm,width=4.5cm,keepaspectratio]{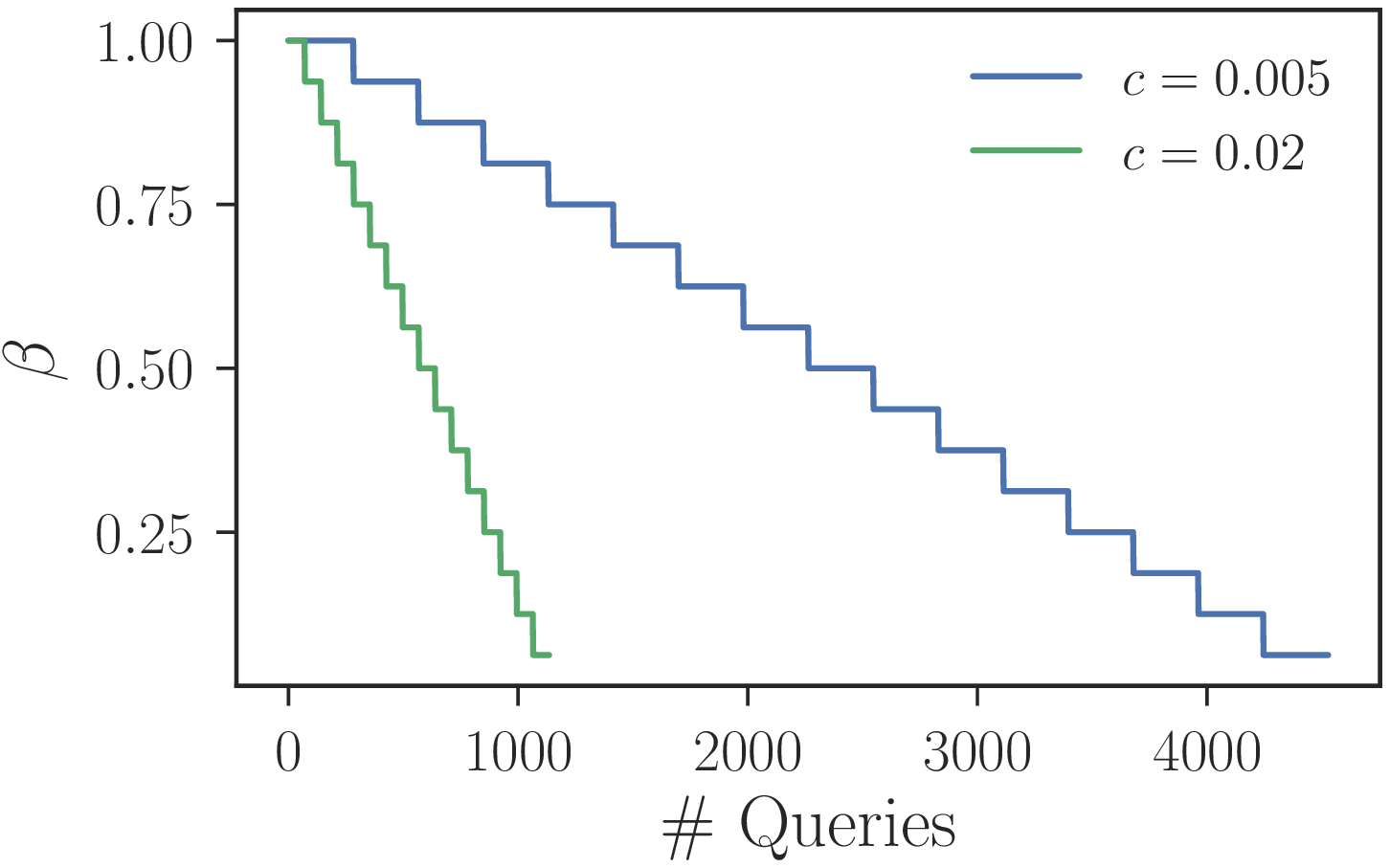}\\
\end{tabular}
\caption{(Left) Error vs. convergence $c$; (right) diminishing probability of buffering $\beta$ as more query spaces are \textit{known}.}
\label{fig:sensitivity-c}
\end{center}
\vspace{-0.25cm}
\end{figure}
\FloatBarrier
Figure \ref{fig:sensitivity-c}(right) shows the diminishing probability of entering the buffering mode $\beta = P(G > h)$ building upon our discussion of slowly converging the system into an Offline mode. 
As more queries are processed across varying query spaces 
then our system is incrementally learning the whole data space. 
At a certain point, all query subspaces will be \textit{known} along 
with their representatives. Thus, the probability 
of entering the buffering mode due to potentially unknown query distribution 
reduces to zero almost surely. 
We provide an experiment in which there is a predefined 
fixed number of Query Spaces (QS) $K=16$, $QS = \{QS_1, \ldots, QS_{K}\}$. 
Queries are generated randomly in a sequence from one QS to another, each time learning $QS_{k-1}$. Thus, given this fixed number of QSs and set $h, c$, 
the probability of entering buffering mode can be approximated by 
$\beta = P(G > h) = 1 - \frac{l}{K}$, where $l \leq K$ 
denotes the number of known QS so far. 
Liaising this with 
Figure \ref{fig:sensitivity-c}, we observe that 
the probability reduces in a step-wise manner
tending to zero when $l \to K$. 
For a relatively high $c$ value, 
the rate of convergence to the offline mode 
becomes faster but with a higher error as witnessed by the previous experiment. 
As for parameter $h$, a low value indicates 
smaller tolerance when estimating errors and vice versa. This might force the system to adapt when not needed. 
Thus, it is domain appropriate to hyper-tune the parameter accordingly; hyper-tuning $h$ is part of 
our on-going research. However, we have found the proposed heuristic 
of $3\sigma\leq h \leq 5\sigma$ to work well empirically.
\section{Related Work}
Our work is related to prior work in analytical-query processing and in applied ML research communities and to prior work focusing on the benefits of the query-driven approach in analytical query processing and tuning \cite{anagnostopoulos2015learning,anagnostopoulos2017query,ma2018query,sattler2003quiet}. Analytical queries nowadays are executed over underlying systems that provide either exact answers \cite{melnik2010dremel,thusoo2009hive} or approximate answers  \cite{park2018verdictdb,agarwal2013blinkdb,park2017database,kandula2016quickr,hellerstein1997online} working over large big data clusters in DCs/CS requiring several orders of magnitude longer query response times. The contributions in this work are largely complementary to all this work. Specifically, during the training phase and in the adapting/buffering phase the system proposed here can be supported either by an exact or an approximate query processing engine. In addition, what makes our solution different is that it can be stored locally on an analyst's device as it has low storage overhead and also requires no communication to the cluster. 

Query-driven models are largely being deployed for both aggregate estimation \cite{anagnostopoulos2017query,anagnostopoulos2015learning} and for hyper-tuning \cite{van2017automatic} database systems. Unlike \cite{anagnostopoulos2017query,anagnostopoulos2015learning} our focus is on a wide variety of aggregate operators and not just \texttt{COUNT} for selectivity estimation. Furthermore, we address the crucial problem of detecting query pattern changes and adapting to them, which (to our knowledge) has not been addressed in this context before. Hence,our framework can be leveraged by all query-driven implementations in cases of dynamic workloads that are non-stationary.
%This brings our approach to the exploration and analytics stack.  

Moreover, concept drift adaptation 
is well understood \cite{tsymbal2004problem,gepperth2016incremental,ditzler2015learning,elwell2011incremental}, mostly dealing with classification tasks, where classifiers adapt to new classes. We adapt concept drift to   
query-driven analytical processing, relying on \textit{explicit partitioning} \cite{gepperth2016incremental},
ensuring it avoids destructive forgetting given that the accuracy for the previously learned query patterns will not degrade. It is also favorable given our initial off-line design which already uses partitioning for clustering the query patterns and learning local models in given sub-spaces. Our work contributes 
with monitoring and detecting real-time query patterns change based on \textit{approximating} the prediction error, which differentiates with 
the previous concept drift methods by measuring the actual error;
evidently, this is not applicable in our case. Finally, we propose a novel reciprocity-driven adaptation mechanism in which we set a mechanism deciding when a new model should be trained engaging the knowledge derived from other possibly changing models in the CS.
\section{Conclusions}
In this work we contribute a novel framework for adapting trained models under concept drift. We focus on models used for estimating  analytical query answers efficiently and accurately, however we note that the framework is applicable in other domains as well.
The contributions centre on a novel suit of ML models, which mine past and new queries and incrementally build models over quantized query-spaces using a vectorial representation. The described mechanisms (ADM and CDM) bear the ability to adapt under changing analytical workloads, while maintaining high accuracy of estimations. As shown by our evaluation (using real and synthetic datasets), the proposed approach enjoys high accuracy (well below 10\% relative error) across all aggregate operators, with low response times (well below a millisecond) and low footprint and training-time overheads. The contributed adaptability mechanism is able to detect changes using estimated errors and swiftly adapt. Furthermore, as more queries are processed our system has the potential to reach global convergence as no more query patterns remain undiscovered. This can significantly reduce unnecessary communication to cloud providers thus reduce network load and monetary costs.

\section{Acknowledgement}
Fotis Savva is funded by an EPSRC scholarship, grant number 1804140. Christos Anagnostopoulos and Peter Triantafillou were also supported by EPSRC grant EP/R018634/1

\FloatBarrier
\bibliographystyle{abbrv}
\balance
\bibliography{bibliography}
\newpage
\section{Appendix}
\subsection{Real Datasets \& Workloads}
\textbf{Crimes:} 
For constructing the workload over the \textbf{Crimes} dataset we initially sampled 10K points and obtained the mean and standard deviation for the attributes \textit{X\_Coordinate},\textit{Y\_Coordinate}. By setting the obtained statistics as parameters for a multivariate normal distribution we generated a number of \textit{Cluster} points ($=5$). Centered on each one of those cluster points we constructed multivariate normal distributions with a fraction ($=0.01$) of the original standard deviation. Based on this we generated a number of points ($=10K$) for around each cluster point. Based on those points we constructed queries covering a random range across both the \textit{X/Y Coordinates}. For each point a fraction of the original std was used to construct ranges : $\text{x\_range} = \frac{\text{x\_std}}{2}\text{rand(0,1)}$, where \textit{x\_std} is the original standard deviation, same approach was followed for the \textit{y\_range}. The queries filtered the complete dataset based on (where pt is the point) :
\begin{multline}
\alpha_{X\_Coordinate} \geq \text{pt}-\text{x\_range} \land \alpha_{X\_Coordinate} \leq \text{pt}+\text{x\_range}\\
\land \alpha_{Y\_Coordinate} \geq \text{pt}-\text{y\_range} \land \alpha_{y\_Coordinate} \leq \text{pt}+\text{y\_range}
\end{multline}
On the filtered dataset generated by each query with varying cardinality we extracted basic statistics like \texttt{COUNT, AVG, SUM} on attributes that would make sense (\textit{Beat} - avg, \textit{Arrest} - sum). 

\textbf{Sensors}: 
For Sensors we obtained the mean and standard deviation of the temporal dimension after encoding it and normalizing it. The min/max statistics were also obtained. We then generated a number of queries ($=50K$) with range equal to a fraction of the complete distance between the max/min ($0.2\times(\text{max}-\text{min})$. The center of the queries was randomly generated by a normal distribution with parameters equal to the obtained mean and half original standard deviation. The same processes as in Crimes was used to filter the complete dataset and extract statistics on attributes that made sense (\textit{Temperature}-avg, \textit{Light}-sum).

\subsection{Synthetic Data Generation}
Specifically, we generated datasets with varying dimensionality $[10, 20, 50, 100]$ to simulate data with large number of columns. Each column contains numbers generated from a uniform distribution $\mathcal{U}[0, 10^{6}]$, therefore each point is $\mathbf{x} \in [0, 10^{6}]^d$ where $d \in [10, 20, 50, 100]$. We then generated a number of query workloads with a varying number of selected columns to be restricted by predicates $p \in [2, 5, 10]$ (Increasing predicates even further resulted in no tuples returned) according to the dimensionality of each of the synthetic datasets. For instance, for a dataset with $d=10$ the resulting workloads were two, with the selected columns being $(2, 10)$. Each query point for those two datasets is then $\mathbf{q} \in [0, 10^{6}]^{2p}$.

We also describe the query generation process in Algorithm \ref{alg:query-gen}. We first set the number of queries $N$ and number of predicates and columns in our datasets $\mathcal{D}$, $\mathcal{P}$. The retrieved range size $r$ is obtained by a Normal distribution $\mathcal{N}(\text{sel}\cdot10^6, 0.01\cdot10^6)$ with the mean being a fraction of the complete range defined by selectivity ration $sel=0.5$ to ensure enough tuples are returned as with an increasing number of predicates less tuples are returned. We then loop through over the parameters issued and generated queries using the provided algorithms. Where the $random$ function returns a random vector of size $p$ over the given range, and $randomBit$ generates a random bitmask for the selection of columns taking part in the query. Then the lower-bound and upper-bound for the given predicates is set and we execute the query over the restricted space given by the bounds over the selected columns by $b$ using dataset $d$.

\begin{algorithm}
\SetKwInOut{Input}{Input}
\SetKwInOut{Output}{Output}
 \caption{Generating Queries}
 \label{alg:query-gen}
 \Input{$N$,\# queries, $\mathcal{D} = [10, 20, 50, 100]$, \# columns}
 \Input{$\mathcal{P}= [2, 5, 10]$, \# predicates (selected columns)}
 \Input{$r$, range size retrieved from histogram to ensure similar selectivities}
 \For {$d\in \mathcal{D}$}{
 	\For {$p \in \mathcal{P}$}{
    	\For {$i= 1 \textbf{ to } N$}{
        	$\mathbf{z} \gets random(\mathcal{N}(10^6, 100), p)$ \;
            $\mathbf{b} \gets randomBit([0, d], p)$ \; 
            $\mathbf{lb} \gets \mathbf{z}-\frac{r}{2}$ \;
            $\mathbf{ub} \gets \mathbf{z}+\frac{r}{2}$ \;
            $\mathbf{q}_i^{(m)} \gets [\mathbf{lb}, \mathbf{ub}]$\;
            $y_i \gets executeAggregates(\mathbf{q}_i^{(m)}, d)$ \;
            $\mathbf{q}_i \gets (\mathbf{q}_i^{(m)},y)$ \;
            $\mathcal{W}_{d,p} \cup \{\mathbf{q}_i\}$\;
        }
    }
 }
 \Output{Resulting workloads $\mathbf{W} = (\mathcal{W}_{10,2}, \ldots, \mathcal{W}_{d,p}), \forall d\in \mathcal{D}, \forall p\in \mathcal{P}$}
\end{algorithm}

\end{document}